\documentclass[11pt]{article}
\usepackage[a4paper,left=1.0in,right=1.0in,top=1.1in,bottom=1.3in]{geometry}

\usepackage{microtype}

\usepackage{authblk}
\usepackage{xcolor}
\usepackage{etex,ifthen}
\usepackage{etoolbox}
\usepackage{graphicx}
\usepackage{tikz}
\usepackage{pgfplots}
\pgfplotsset{compat=1.14}
\usepackage{tikz-cd}
\usepackage{url}
\usepackage{calc}
\usepackage[font={small,it}]{caption}
\usepackage{complexity}
\usepackage{alphalph}
\usepackage{amsthm}
\usepackage{thmtools}
\usepackage{eqparbox}
\usepackage{afterpage}
\usepackage{amsfonts,amsmath,amssymb}
\usepackage{algorithm}
\usepackage{algpseudocode}
\usepackage{bm}
\usepackage{framed}
\usepackage{listings}
\usepackage{hyperref}
\usepackage{verbatim}
\usepackage[all]{xy}

\usetikzlibrary{chains,fit}
\usetikzlibrary{shapes.geometric}
\usetikzlibrary{matrix,positioning,calc}
\usetikzlibrary{decorations.markings}
\usetikzlibrary{decorations.pathmorphing}
\tikzstyle{vertex}=[circle, draw, inner sep=0pt, minimum size=6pt]
\tikzstyle{vertbox}=[draw, inner sep=0pt, minimum size=8pt]
\newcommand{\vertex}{\node[vertex]}

\newcommand{\oset}[3][0ex]{%
	\mathrel{\mathop{#3}\limits^{
			\vbox to#1{\kern-2\ex@
				\hbox{$\scriptstyle#2$}\vss}}}}

\newcommand{\cA}{\mathcal{A}}
\newcommand{\cI}{\mathcal{I}}
\newcommand{\cU}{\mathcal{U}}

\newcommand{\fI}{\mathfrak{I}}
\newcommand{\fP}{\mathfrak{P}}
\newcommand{\fX}{\mathfrak{X}}
\newcommand{\row}{\mathfrak{r}}
\newcommand{\pseudo}{\mathfrak{p}}

\newcommand{\lefty}{\mathsf{left}}
\newcommand{\righty}{\mathsf{right}}

\newcommand{\Gright}{G_{\operatorname{right}}}

\newcommand{\Gnew}{G_{\operatorname{new}}}
\newcommand{\Gdiag}{G_{\operatorname{diag}}}

\newcommand{\GBFS}{G_{\operatorname{BFS}}}
\newcommand{\Gset}{G_{\operatorname{set}}}

\newcommand{\xin}{x_{\operatorname{in}}}
\newcommand{\xout}{x_{\operatorname{out}}}
\newcommand{\xiin}{x_{i,{\operatorname{in}}}}
\newcommand{\xiout}{x_{i,{\operatorname{out}}}}
\newcommand{\inn}{\mathsf{in}}
\newcommand{\out}{\mathsf{out}}

\newcommand{\xprev}{x^{\operatorname{pre}}}

\newcommand{\sssp}{\operatorname{\mathbf{sssp}}}
\newcommand{\optsol}{\operatorname{OPT}}

\newcommand{\greedyG}{P_G^{\operatorname{gr}}}
\newcommand{\greedyX}{P_X^{\operatorname{gr}}}
\newcommand{\greedyY}{P_Y^{\operatorname{gr}}}
\newcommand{\greedyI}{P_I^{\operatorname{gr}}}
\newcommand{\east}{P_X^{\operatorname{E}}}
\newcommand{\west}{P_X^{\operatorname{W}}}
\newcommand{\north}{P_Y^{\operatorname{N}}}
\newcommand{\neast}{P_G^{\operatorname{NE}}}
\newcommand{\nwest}{P_G^{\operatorname{NW}}}
\newcommand{\myparagraph}[1]{\noindent\textbf{#1\ }}

\declaretheorem[numberlike=equation]{Theorem}
\declaretheorem[numberlike=equation]{Lemma}

\declaretheoremstyle[bodyfont=\it,qed=$\lozenge$]{defstyle} 

\declaretheorem[numberlike=equation,style=defstyle]{Definition}
\declaretheorem[numberlike=equation]{Claim}

\hypersetup{
	colorlinks,
	linkcolor=red,
	citecolor=blue,
	urlcolor=green!50!black
}

\title{Minimizing Branching Vertices in Distance-preserving Subgraphs}

\author{Kshitij Gajjar\thanks{\texttt{kshitij.gajjar@tifr.res.in}}~ and Jaikumar Radhakrishnan\thanks{\texttt{jaikumar@tifr.res.in}}}
\affil{Tata Institute of Fundamental Research, Mumbai}

\begin{document}
	
	\maketitle

	
	

	\begin{abstract}
		
		It is $\NP$-hard to determine the minimum number of branching vertices needed in a single-source distance-preserving subgraph of an undirected graph. We show that this problem can be solved in polynomial time if the input graph is an interval graph.

		In earlier work, it was shown that every interval graph with $k$ terminal vertices admits an all-pairs distance-preserving subgraph with $O(k\log k)$ branching vertices~\cite{jaikumar}. We consider graphs that can be expressed as the strong product of two interval graphs, and present a polynomial time algorithm that takes such a graph with $k$ terminals as input, and outputs an all-pairs distance-preserving subgraph of it with $O(k^2)$ branching vertices. This bound is tight.
	\end{abstract}

	
	

	\section{Introduction}
	
	Distance-preserving minors were introduced by Krauthgamer and Zondiner~\cite{KraZon}. They showed that every undirected graph with $k$ terminals admits a distance-preserving minor with $O(k^4)$ vertices and edges, and presented a planar graph with $k$ terminals for which every distance-preserving minor has $\Omega(k^2)$ vertices.
	
	Distance-preserving subgraphs are closely related to distance-preserving minors. In fact, for many graph classes, distance-preserving minors are obtained by first constructing distance-preserving subgraphs and contracting edges adjacent to vertices of degree $2$.
	
	In such cases, the size of the distance-preserving minor depends solely on the number of vertices of degree $3$ or more in the distance-preserving subgraph. We call such vertices \emph{branching vertices}. It is therefore natural to minimize the number of branching vertices while constructing distance-preserving subgraphs.
	
	Gajjar and Radhakrishnan~\cite{jaikumar} consider the problem of constructing distance-preserving subgraphs of interval graphs and showed that every interval graph with $k$ terminals admits a distance-preserving subgraph with $O(k\log k)$ branching vertices. They also showed that this bound is tight. We later present a setting (borrowed from~\cite[Section 1.2]{jaikumar}) in which interval graphs arise naturally for this problem.
	
	In this paper, we consider the algorithmic question of finding distance-preserving subgraphs with the fewest branching vertices. This problem was already shown to be $\NP$-complete for general graphs~\cite{jaikumar}. We observe that the same proof shows that even the single-source version of the problem is $\NP$-complete for general graphs (see~\autoref{nphard}). We show the following result.
	
	\begin{Theorem} [Single-source interval graphs] \label{main1}
	There exists a polynomial time algorithm that, given an interval graph $G$ with a source $s$ and $k$ terminals as input, constructs a shortest path tree for $G$ with the minimum number of branching vertices.
	\end{Theorem}
	
	As stated above, the upper bound for interval graphs is significantly better than the lower bound for general graphs. We ask if similar better upper bounds can be established for super-classes of interval graphs. Since interval graphs are precisely the graphs of boxicity one, a natural candidate is graphs of boxicity two (intersection graphs of axis-parallel rectangles). However, it can be shown that there are graphs of boxicity two with $k$ terminals that require $\Omega(k^4)$ branching vertices in any distance-preserving subgraph (see~\autoref{boxicity}).
	
	A sub-class of graphs of boxicity two is bi-interval graphs. These are intersection graphs of axis-parallel rectangles, where the rectangles arise from the cross product of two families of intervals (see~\autoref{def:biinterval}).
	
	\begin{Theorem} [All-pairs bi-interval graphs] \label{main2}\ 
		\begin{enumerate}
			\item[(a)] (Upper bound) There is a polynomial time algorithm that, given a bi-interval graph with $k$ terminals as input, produces a distance-preserving subgraph of it with $O(k^2)$ branching vertices.
			\item[(b)] (Lower bound) For every $k\geq 4$, there is a bi-interval graph $\Gdiag$ on $k$ terminals such that every distance-preserving subgraph of $\Gdiag$ has $\Omega(k^2)$ branching vertices.
		\end{enumerate}
	\end{Theorem}
	
	\subsection{Related Work}
	
	The problem of constructing small distance-preserving subgraphs bears close resemblance to several well-studied problems in graph algorithms: graph compression~\cite{Feder}, graph spanners~\cite{Bodwin, Coppersmith, Pel}, Steiner point removal~\cite{Filtser, AnupamGupta, Kamma}, graph homeomorphisms~\cite{fortune, lapaugh}, graph contractions~\cite{Daub}, graph sparsification~\cite{goranci, spielman}, etc.
	
	Note that there are several different notions of distance-preserving subgraphs. Our notion of distance-preserving subgraphs is different from that used by Djokovi\'c~\cite{djok} (see also~\cite{chepoi}), Nussbaum \emph{et al.}~\cite{nussbaum}, Yan \emph{et al.}~\cite{yan}, or Sadri \emph{et al.}~\cite{sadri}.
	
	\subsection{Our Techniques}
	
	In this section, we explain the main ideas behind our algorithms and proofs.
	
	Our first result concerns the construction of single-source distance-preserving subgraphs in interval graphs with the minimum number of branching vertices. Constructing a single-source shortest path tree in an interval graph is straightforward: the simple greedy algorithm does the job. However, this method is not guaranteed to produce a tree with the \emph{minimum} number of branching vertices. We do not know of a simple modification of the greedy method that helps us minimize the number of branching vertices. We instead proceed as follows. Using a breadth-first search, we partition the graph into layers based on their distance from the source vertex. Within a layer, we order vertices based on the position of the end points of their respective intervals. We observe that with respect to this ordering an optimal solution can be assumed to have a very special structure: some vertices are directly connected by paths to the source, and the others via a single special vertex. Once this structure is established, a combination of network flows and dynamic programming allows us to determine the optimal solution. Since finding a single-source shortest path tree in interval graphs is simple, it is not clear why we need a somewhat sophisticated solution if we need to minimize the number of branching vertices. We describe our solution in detail (see~\autoref{main1again}) assuming that the source vertex is the leftmost interval. This case already contains most of the ideas; the algorithm for the general case requires some more analysis.
	
    Our second result concerns a 2D generalization of interval graphs. In earlier work~\cite{jaikumar}, the worst-case number of branching vertices for distance-preserving subgraphs of interval graphs with $k$ terminals was determined to be $\Theta(k \log k)$. In particular, the lower bound of $\Omega(k \log k)$ was obtained by considering certain interval graphs with regularly placed intervals. A natural 2D generalization of such graphs would be regularly placed rectangles, a special case of bi-interval graphs (\autoref{king}). We present an example (\autoref{main2bagain}) where such graphs with $k$ terminals require $\Omega(k^2)$ branching vertices. To show that this is tight, we observe that shortest paths in such graphs can be constructed with two components: (i) ones that proceed \emph{diagonally}, (ii) ones that proceed \emph{horizontally} or \emph{vertically} (we call such paths \emph{straight} paths). We first generate the diagonal paths from all terminals; with some care one can argue that these paths meet only at $O(k^2)$ vertices. However, we need to complete these partial paths by adding straight segments. We observe that this second phase of our plan can be mapped to a certain problem for constructing distance-preserving subgraphs in interval graphs. By borrowing some ideas from earlier work~\cite{jaikumar}, we show how these problems can be solved efficiently in such a way that the total number of branching vertices remains $O(k^2)$.

	\subsection{An Example} \label{sec:shipping}
	
	The following example (taken from~\cite{jaikumar}) illustrates a setting where distance-preserving subgraphs of interval graphs arise naturally when considering a shipping problem.
	
	 A freight container needs to be delivered from  seaport $X$ to seaport $Y$, but there is no ship that travels from $X$ to $Y$. In such cases, the container is typically first transported from port $X$ to a central hub $H$, and transferred through a series of ships arriving at $H$ until it is finally picked up by a ship that is destined for port $Y$. Thus, the container reaches its final destination via some ``intermediate'' ships at port $H$~\footnote{The container cannot be left at the warehouse/storage unit of port $H$ itself beyond a certain limited period of time.}.
	
	However, there is a cost associated with transferring the container from one ship to another. There is also an added cost if an intermediate ship receives containers from multiple ships, or sends containers to multiple ships. Thus, given the docking times of ships at $H$, and a small subset of these ships that require a transfer of containers between each other, our goal is to devise a transfer strategy that meets the following objectives: (i) Minimize the number of transfers for each container; (ii) Minimize the number of ships that have to deal with multiple transfers. 
	
	Representing each ship's visit to the port as an interval on the time line, this problem can be modelled using distance-preserving subgraphs of interval graphs. In this setting, a shortest path from an earlier interval to a later interval corresponds to a valid sequence of transfers across ships that moves forward in time. Objective (i) corresponds to minimizing pairwise distances between terminals, and objective (ii) corresponds to minimizing the number of branching vertices in the subgraph.
	
	In previous work~\cite{jaikumar}, it was shown that $O(k \log k)$ branching vertices are sufficient for preserving pairwise distances between $k$ terminals, and that this is the best possible bound. In the next section, we focus on the problem of determining this number exactly. For general graphs this problem is $\NP$-complete, even when restricted to one source. We do not know if the all-pairs version of the problem is $\NP$-complete for interval graphs. We show that the problem can be solved efficiently for interval graphs when there is only one source.

	
	

    \section{Single-source Distance-preserving Subgraphs}
    
    \begin{figure}
    \begin{center}
    \begin{tikzpicture} [thick,every node/.style={},every fit/.style={ellipse,draw,inner sep=-12pt,text width=2cm}]

    \begin{scope}[xshift=1cm,yshift=-0.5cm,start chain=going below,node distance=7mm]
    \foreach \i in {6,7,...,9}
        \node[on chain] (d1\i) [] {};
    \end{scope}

    \begin{scope}[xshift=3.5cm,yshift=-0.5cm,start chain=going below,node distance=7mm]
    \foreach \i in {6,7,...,9}
        \node[on chain] (d2\i) [] {};
    \end{scope}

    \begin{scope}[xshift=6cm,yshift=-0.5cm,start chain=going below,node distance=7mm]
    \foreach \i in {6,7,...,9}
        \node[on chain] (d3\i) [] {};
    \end{scope}

    \begin{scope}[xshift=8.5cm,yshift=-0.5cm,start chain=going below,node distance=7mm]
    \foreach \i in {6,7,...,9}
        \node[on chain] (d4\i) [] {};
    \end{scope}

    \node [blue!25,fit=(d16) (d19),label=above:$R_1$] {};
    \node [blue!25,fit=(d26) (d29),label=above:$R_2$] {};
    \node [blue!25,fit=(d36) (d39),label=above:$R_3$] {};
    \node [blue!25,fit=(d46) (d49),label=above:$R_4$] {};

    \vertex(aa) at (-1,-2) [label=left:$s$, blue, fill=blue] {};
    \draw[ultra thin](aa)--(1,-1.85);
    \draw[ultra thin](aa)--(1,-3.10);
    \draw[ultra thin](aa)--(1,-1.30);
    \draw[ultra thin](aa)--(1,-0.75);
    
    \node[] at (2.25,-2.45) {$\cdots$};
    \node[] at (4.75,-2.45) {$\cdots$};
    \node[] at (7.25,-2.45) {$\cdots$};
    
    \node[] at (2.25,-1.30) {$\cdots$};
    \node[] at (4.75,-1.30) {$\cdots$};
    \node[] at (7.25,-1.30) {$\cdots$};

    \vertex at (1,-0.75) [minimum size=4pt, gray, fill=gray] {};
    \vertex at (1,-1.30) [minimum size=4pt, gray, fill=gray] {};
    \vertex at (1,-1.85) [minimum size=4pt, gray, fill=gray] {};
    \node[] at (1,-2.35) {$\vdots$};
    \vertex at (1,-3.10) [minimum size=4pt, gray, fill=gray] {};

    \vertex at (3.5,-0.75) [minimum size=4pt, gray, fill=gray] {};
    \vertex(v2) at (3.5,-1.30) [minimum size=4pt, gray, fill=gray, label=left:\footnotesize$v_2$] {};
    \vertex at (3.5,-1.85) [minimum size=4pt, gray, fill=gray] {};
    \node[] at (3.5,-2.35) {$\vdots$};
    \vertex(v1) at (3.5,-3.10) [minimum size=4pt, gray, fill=gray, label=left:\footnotesize$v_1$] {};
    
    \vertex(x1) at (6,-0.75) [minimum size=4pt, gray, fill=gray] {};
    \vertex(x2) at (6,-1.30) [minimum size=4pt, gray, fill=gray] {};
    \vertex(x3) at (6,-1.85) [minimum size=4pt, gray, fill=gray] {};
    \node[] at (6,-2.35) {$\vdots$};
    \vertex(x4) at (6,-3.10) [minimum size=4pt, gray, fill=gray] {};
    
    \vertex at (8.5,-0.75) [minimum size=4pt, gray, fill=gray] {};
    \vertex at (8.5,-1.30) [minimum size=4pt, gray, fill=gray] {};
    \vertex at (8.5,-1.85) [minimum size=4pt, gray, fill=gray] {};
    \node[] at (8.5,-2.35) {$\vdots$};
    \vertex at (8.5,-3.10) [minimum size=4pt, gray, fill=gray] {};
    
    \draw[ultra thick,red](v1)--(x4);
    \draw[ultra thick,red](x1)--(v2)--(x3)--(v2)--(x4);
    
    \end{tikzpicture}
    \caption{An instance of $\GBFS$. The source vertex $s$ is the leftmost interval in $G$. The layer $R_i$ is the set of vertices at a distance of $i$ from $s$, and the vertices within each layer are arranged in increasing order from bottom to top. For instance, $v_1$ is below $v_2$ in $R_2$, which means that the right end point of the interval $v_1$ is to the left of the right end point of the interval $v_2$. Thus, the neighbourhood of $v_1$ in $R_3$ is a subset of the neighbourhood of $v_2$ in $R_3$.}
    \label{fig:bfsrep}
    \end{center}
    \end{figure}
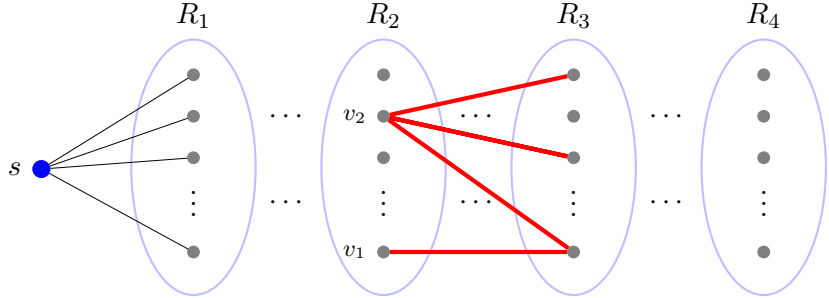
    
    In this section, we investigate the computational complexity of finding distance-preserving subgraphs with the minimum number of branching vertices. We call such subgraphs \emph{optimal} distance-preserving subgraphs. Let $\cA$ be the following decision problem.
    
    \begin{description}
	    \item{Input:} A graph $G$ and a positive integer $m$.
	    \item{Output:} Yes, if there is a single-source distance-preserving subgraph of $G$ with at most $m$ branching vertices; No, otherwise.
	\end{description}
    
    The main theorem for this section is the following, which is simply~\autoref{main1} restated over here.
    
	\begin{Theorem} \label{main1again}
	The task $\cA$ can be carried out in polynomial time, if $G$ is an interval graph.
	\end{Theorem}
	
	We shall see a proof~\autoref{main1again} shortly. First, let us investigate what happens in the general setting.
	
    \subsection{General Graphs}
	
	In this section, we will prove that it is $\NP$-complete to find a distance-preserving subgraph of a general graph with the minimum number of branching vertices.
    
    \begin{Theorem} \label{nphard}
    The decision problem $\cA$ is $\NP$-complete.
    \end{Theorem}
    
    \begin{proof} In~\cite[Section 2.1]{kshitij}, it is shown that finding an optimal all-pairs distance-preserving subgraph is $\NP$-hard for general graphs. We will show that the same proof carries through for the single-source case as well.
    
    Consider the graph $\Gset$ in~\cite[Theorem 6]{kshitij}. Let $S=\{t_1\}$ be the source vertex and let $T=\cU\cup\{t_0,t_1\}$ be the set of target vertices. Observe that the proof essentially shows that even preserving distances from $t_1$ to $T$ is $\NP$-hard. This completes the proof.
	\end{proof}
    
	Since there is no hope of solving the problem for the general case (unless $\P=\NP$), we move on to interval graphs.
	
	\subsection{Representing Interval Graphs}
	
	Every interval graph has two representations: one using vertices and edges (called the \emph{graph} representation), and one using a set of intervals on the real line (called the \emph{interval} representation).
	
	Let $G$ be an interval graph and let $\fI_G$ be the interval representation of $G$. For every vertex $v\in V(G)$, let $\lefty(v)$ and $\righty(v)$ be the left and right end points, respectively, of its corresponding interval in $\fI_G$. For simplicity, we assume that all the end points of the intervals have distinct values in $\fI_G$. Define a relation ``$\prec$'' on the vertices of $G$ such that $v_1\prec v_2$ in $G$ if and only if $\righty(v_1)<\righty(v_2)$ in $\fI_G$. In other words, the intervals are ordered according to their right end points. The relation ``$\preceq$'' is similarly defined. When talking about interval graphs, we use the terms interval and vertex interchangeably, and when talking about bi-interval graphs, we use the terms rectangle and vertex interchangeably.
	
	It is well known that one method of constructing shortest paths in interval graphs is the following \emph{greedy} algorithm. Suppose we need to construct a shortest path from interval $u$ to interval $v$ in $G$ (assume $u \prec v$). The greedy algorithm starts at $u$. The next vertex on the greedy shortest path from $u$ to $v$ is the interval with the maximum $\righty$ value of all the intervals that overlap with $u$. In this way, each step of the greedy algorithm chooses the next interval that overlaps with the current interval and has the maximum $\righty$ value. It stops as soon as the current interval overlaps with $v$. It is easy to prove that the path thus obtained is a shortest path from $u$ to $v$. Let $\greedyG(u,v)$ be the shortest path produced by this greedy algorithm.
	
	\subsection{Interval Graphs with the Leftmost Interval as the Source Vertex}
	
	We now present a polynomial time algorithm that takes an interval graph as input and outputs an optimal distance-preserving subgraph of it. Here is a brief outline of our algorithm. There are two parts to our proof. We first show that we may restrict attention to solutions that have a special structure. To find the best solution with this structure, we show a natural decomposition of the optimal solution into two parts, where the first part ensures that $s$ has shortest paths to a subset of terminals; the second part consists of one or two instances of the same problem. This decomposition allows one to use dynamic programming to solve the problem. In this section, we will restrict attention to interval graphs where $s$ is the \emph{leftmost} interval. The modifications needed to solve the general case, where $s$ may appear in the middle, are described later.
	
	\myparagraph{BFS:} Consider the breadth-first search (BFS) tree rooted at $s$. The tree naturally partitions the vertices into layers $R_0,R_1,\ldots$, where a vertex is placed in layer $i$ if its shortest distance from $s$ is $i$. We arrange the layers from left to right, where $R_0=\{s\}$ is the leftmost. For a vertex $v$, let $\ell_v$ be the layer in which $v$ belongs. We further arrange vertices within a layer from bottom to top. Suppose $v_1,v_2 \in R_i$. Then we place $v_1$ below $v_2$ in $R_i$ if and only if $v_1\prec v_2$ (see~\autoref{fig:bfsrep}). The edges of $G$ connect vertices that are in the same layer or in adjacent layers. We discard all edges whose end points fall within the same layer, for such edges cannot be in any shortest path from $s$; we direct the remaining edges away from $s$ (that is, from left to right). We refer to this graph as $\GBFS$.

    \myparagraph{Decomposition:} Let us consider the problem of computing the optimal solution rooted at a vertex $v$ (in layer $i$, say) that provides shortest paths from $v$ to all terminals in layers $R_{i+1}, R_{i+2}, \ldots$; we refer to this problem as $\sssp(v)$, where $\sssp$ stands for ``single-source shortest path''. We will show that there is an optimal solution $\optsol_v$ rooted at $v$ with the following structure. There is a vertex $w$ (say in layer $R_j$, $j >i$) such that $\optsol_v$ is the union of a solution for $\sssp(w)$ and some paths from $v$ that include all terminals in layers $R_0,R_1, \ldots, R_i$ and $w$. These paths have only $v$ in common, so they do not contribute any branching vertices except perhaps at $v$. This decomposition naturally gives a dynamic programming solution for our problem. To describe our dynamic programming solution based on this observation, we will use the following definitions.

    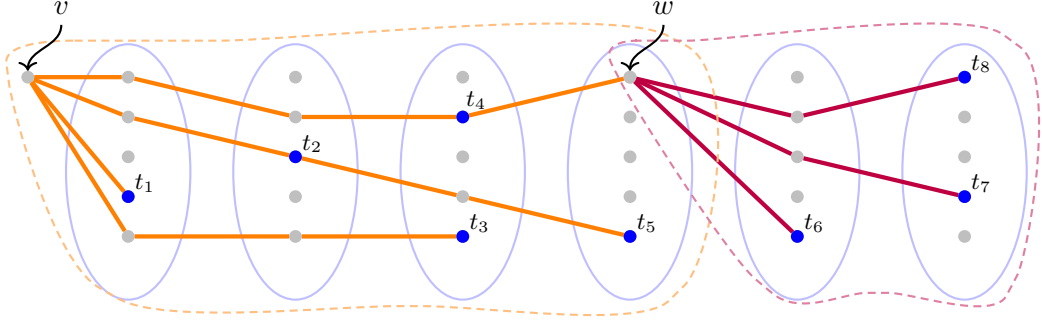
\begin{figure}
    \begin{center}
    \begin{tikzpicture} [scale=0.88,thick,every node/.style={},every fit/.style={ellipse,draw,inner sep=-12pt,text width=2cm}]

    \begin{scope}[xshift=1cm,yshift=-0.5cm,start chain=going below,node distance=7mm]
    \foreach \i in {6,7,...,9}
        \node[on chain] (d1\i) [] {};
    \end{scope}

    \begin{scope}[xshift=3.5cm,yshift=-0.5cm,start chain=going below,node distance=7mm]
    \foreach \i in {6,7,...,9}
        \node[on chain] (d2\i) [] {};
    \end{scope}

    \begin{scope}[xshift=6cm,yshift=-0.5cm,start chain=going below,node distance=7mm]
    \foreach \i in {6,7,...,9}
        \node[on chain] (d3\i) [] {};
    \end{scope}

    \begin{scope}[xshift=8.5cm,yshift=-0.5cm,start chain=going below,node distance=7mm]
    \foreach \i in {6,7,...,9}
        \node[on chain] (d4\i) [] {};
    \end{scope}
    
    \begin{scope}[xshift=11cm,yshift=-0.5cm,start chain=going below,node distance=7mm]
    \foreach \i in {6,7,...,9}
        \node[on chain] (d5\i) [] {};
    \end{scope}
    
    \begin{scope}[xshift=13.5cm,yshift=-0.5cm,start chain=going below,node distance=7mm]
    \foreach \i in {6,7,...,9}
        \node[on chain] (d6\i) [] {};
    \end{scope}

    \node [blue!25,fit=(d16) (d19)] {};
    \node [blue!25,fit=(d26) (d29)] {};
    \node [blue!25,fit=(d36) (d39)] {};
    \node [blue!25,fit=(d46) (d49)] {};
    \node [blue!25,fit=(d56) (d59)] {};
    \node [blue!25,fit=(d66) (d69)] {};
    
    \def \orig {-0.75}
    \def \vgap {0.6}
    \def \hgap {2.5}
    \def \vtag {0.2}
    \def \htag {0.25}
    
    \vertex(vv) at (-0.5,\orig) [minimum size=4pt, lightgray, fill=lightgray] {};

    \foreach \i in {0,...,5}
	    \foreach \j in {0,...,4}
            \vertex(r\i\j) at (\i*\hgap+1,-\j*\vgap+\orig) [minimum size=4pt, lightgray, fill=lightgray] {};
            
    \draw[orange, ultra thick](vv)--(r00)--(r11)--(r21)--(r30);
    \draw[orange, ultra thick](vv)--(r01)--(r12)--(r23)--(r34);
    \draw[orange, ultra thick](vv)--(r03);
    \draw[orange, ultra thick](vv)--(r04)--(r14)--(r24);
    
    \vertex at (1,-3*\vgap+\orig) [minimum size=4pt, blue, fill=blue] {};
    \node at (1+\htag,-3*\vgap+\orig+\vtag) {\footnotesize{$t_1$}};
    
    \vertex at (\hgap+1,-2*\vgap+\orig) [minimum size=4pt, blue, fill=blue] {};
    \node at (\hgap+1+\htag,-2*\vgap+\orig+\vtag) {\footnotesize{$t_2$}};
    
    \vertex at (2*\hgap+1,-4*\vgap+\orig) [minimum size=4pt, blue, fill=blue] {};
    \node at (2*\hgap+1+\htag,-4*\vgap+\orig+\vtag) {\footnotesize{$t_3$}};
    
    \vertex at (2*\hgap+1,-1*\vgap+\orig) [minimum size=4pt, blue, fill=blue] {};
    \node at (2*\hgap+1+\htag-0.05,-1*\vgap+\orig+\vtag+0.05) {\footnotesize{$t_4$}};
    
    \vertex at (3*\hgap+1,-4*\vgap+\orig) [minimum size=4pt, blue, fill=blue] {};
    \node at (3*\hgap+1+\htag,-4*\vgap+\orig+\vtag) {\footnotesize{$t_5$}};
    
    \vertex(ww) at (3*\hgap+1,\orig) [minimum size=4pt, lightgray, fill=lightgray] {};
    
    \draw[purple, ultra thick](ww)--(r41)--(r50);
    \draw[purple, ultra thick](ww)--(r42)--(r53);
    \draw[purple, ultra thick](ww)--(r44);
    
    \vertex at (4*\hgap+1,-4*\vgap+\orig) [minimum size=4pt, blue, fill=blue] {};
    \node at (4*\hgap+1+\htag,-4*\vgap+\orig+\vtag) {\footnotesize{$t_6$}};
    
    \vertex at (5*\hgap+1,-3*\vgap+\orig) [minimum size=4pt, blue, fill=blue] {};
    \node at (5*\hgap+1+\htag,-3*\vgap+\orig+\vtag) {\footnotesize{$t_7$}};
    
    \vertex at (5*\hgap+1,\orig) [minimum size=4pt, blue, fill=blue] {};
    \node at (5*\hgap+1+\htag,\orig+\vtag) {\footnotesize{$t_8$}};
    
    \def \xoff {0.2};
    \def \yoff {-0.2};
    
    \draw [densely dashed, orange!50, rounded corners] plot [smooth]
    coordinates {(0+\xoff,0+\yoff) (2+\xoff,0+\yoff) (9+\xoff,0+\yoff) (9+\xoff,-3.8+\yoff) (5+\xoff,-4+\yoff) (0.5+\xoff,-3.8+\yoff) (-1+\xoff,-0.5+\yoff) (0+\xoff,0+\yoff)};
    
    \draw [densely dashed, purple!50, rounded corners] plot [smooth]
    coordinates {(8.4+\xoff,0+\yoff) (9+\xoff,0+\yoff) (14+\xoff,0+\yoff) (14+\xoff,-3.7+\yoff) (12+\xoff,-3.8+\yoff) (10.3+\xoff,-3.7+\yoff) (8.1+\xoff,-0.5+\yoff) (8.4+\xoff,0+\yoff)};
    
    \node[anchor=north] at (0,\orig+1.3) (v) {$v$};
			\draw (v) edge[out=-90,in=90,->] (vv);
	\node[anchor=north] at (3*\hgap+1.5,\orig+1.3) (w) {$w$};
			\draw (w) edge[out=-90,in=90,->] (ww);
    
    \end{tikzpicture}
    \caption{The structure of an optimal solution in $\GBFS$. Note that $v$ covers all layers up to (and including) that of $w$, and $w$ covers all subsequent layers. This is formally captured by~\autoref{lm:decomposition}.} 
    \label{fig:optsol}
    \end{center}
    \end{figure}

    \begin{Definition} [Cover] \label{coverit}
    Let $v\in V$ and $X\subseteq V$ ($X \neq \emptyset$). We say that $v$ \textbf{covers} $X$ if there exists a tree $\fX$ in the directed graph $\GBFS$ such that
    \begin{enumerate}
        \item $\{v\}\cup X\subseteq V(\fX)$.
        \item $\operatorname{out-deg}_{\fX}(w)\leq 1\qquad \forall\,w\in V(\fX)-\{v\}$.
        \item All leaves of $\fX$ belong to $X$.
    \end{enumerate}
    Thus, $\fX$ is a shortest path tree from $v$ to all the vertices of $X$ such that $\fX$ has no branching vertices other than perhaps $v$ itself. Such a tree $\fX$ is called a cover for $X$ rooted at $v$.
    \end{Definition}
	    
    Fix $v \in R_i, w \in R_j$ ($j>i$). If $\fX$ is a cover for $X:= ((R_{i+1} \cup R_{i+2} \cup \cdots R_j) \cap T) \cup\{w\}$ rooted at $v$, we simply say that $\fX$ is a cover for $w$ rooted at $v$; if such a cover exists then we say that $v$ covers $w$. The decomposition of the solution is formally established below.

    \begin{Lemma} [Decomposition lemma] \label{lm:decomposition} Suppose $v \in V(G)$ is not in in the rightmost layer. Then, there is a vertex $w$ in a layer after $v$'s layer, and an optimal solution $Z$ for $\sssp(v)$ that has the form $A \cup B$, where $A$ is a cover for $w$ rooted at $v$ and $B$ is a solution for the problem $\sssp(w)$. Conversely, suppose $w$ is a vertex in a layer after $v$'s, $A$ is a cover for $w$ rooted at $v$ and $B$ is a solution for $\sssp(w)$. Then, $A \cup B$ is a solution for $\sssp(v)$.
    \end{Lemma}

    \autoref{fig:optsol} provides a nice visualization of the decomposition lemma. We will present the proof of this lemma shortly. First, let us see how to build a dynamic programming solution for the problem based on this lemma. Let $D[v]$ be the minimum number of branching vertices in a solution for $\sssp(v)$; let $D_1[v]$ be the minimum number of branching vertices in a solution for $\sssp(v)$, where $v$ has exactly one outgoing edge (in particular, $v$ is not a branching vertex). Then, $D_1[v]$ and $D[v]$ can be computed using the algorithm $\textsc{FindOpt}(\GBFS)$ presented below.
    
	
	
	\begin{algorithm}
	\begin{algorithmic}[1]
	    \For{each vertex $v$ in $R_r$} \Comment{$R_r$ is the rightmost layer of $\GBFS$.}
            \State $D_1[v]\gets \infty, D[v]\gets 0$
        \EndFor
        \State $i\gets r-1$
        \While{$i\geq 0$} \Comment{$R_0=\{s\}$ is the leftmost layer of $\GBFS$.}
            \For{each vertex $v$ in $R_i$}
                \If{$R_{i+1}$ has no target vertices}
                    \State $D_1[v]\gets \underset{w\in R_{i+1}}{\operatorname{min}} D[w]$
                \ElsIf{$R_{i+1}$ has one target vertex (say $t$)} 
                    \State $D_1[v]\gets D[t]$
                \Else
                    \State $D_1[v]\gets \infty$
                \EndIf
                \State $D[v]\gets
                \operatorname{min}
                \left\{D_1[v],1+\ \underset{\substack{w:\,\ell_w>i;\\ v\text{ covers } w}}{\operatorname{min}}\ D[w]\right\}$
                \State $i\gets i-1$
            \EndFor
        \EndWhile
    \end{algorithmic}\caption{$\textsc{FindOpt}(\GBFS)$}
    \end{algorithm}

	\myparagraph{Correctness:} We will prove by induction (proceeding from right to left) that $D[v]$, as computed by the algorithm $\textsc{FindOpt}()$, is indeed the number of branching vertices in an optimal solution for $\sssp(v)$. The base case, when $v$ is in the rightmost layer, is straightforward. Assume that $D[w]$ has been computed correctly for all vertices $w$ to the right of $v$. Fix an optimal solution $Z$ for $\sssp(v)$ and let $D^*$ be the number of branching vertices in it. By the second part of~\autoref{lm:decomposition}, we have $D[v] \geq D^*$. It remains to show that $D[v] \leq D^*$. If $v$ has out-degree $1$ in $Z$, then let $w$ be its out-neighbour. By induction, $D_1[v] \leq D[w] = D^*$; it follows that $D[v] \leq D_1[v] \leq D^*$. If $v$ has out-degree at least two, then using~\autoref{lm:decomposition}, we obtain a vertex $w$ and a decomposition of $Z=A \cup B$. Then, the number of branching vertices in $B$ is at least $D[w]$, so $D^*= D[w]+1$. Thus, $D[v] \leq D[w] + 1=D^*$. We conclude $D[v]=D^*$, thus completing the proof.
	
	\myparagraph{Efficiency:} Note that to compute $D[v]$, the algorithm needs to determine if $v$ covers $w$. We will verify in the next section that using a polynomial time network-flows algorithm one can efficiently determine if $v$ covers $w$. Since the above algorithm has only two nested loops (the \textbf{while} and \textbf{for} loops in lines 5,6 run that through all vertices, and the loop implicit in line 14), the overall running time of the algorithm is still polynomially bounded in the input size. This confirms that $D[s]$ can be computed efficiently. By storing relevant information along the way in the computation, one can efficiently obtain the optimal solution as well. 
	
	It remains to establish~\autoref{lm:decomposition} and describe how we can determine if $v$ covers $w$, and if it does find a cover.

	\begin{proof}[Proof of~\autoref{lm:decomposition}]
	Fix an optimal solution $\optsol$ rooted at $v$. If $\optsol$ has no branching vertices to the right of $v$, then we take $w$ to be the rightmost terminal. With this choice of $w$, it is easy to verify the claim of the lemma.
	
	Hence suppose that $\optsol$ has a branching vertex to the right of $v$. Suppose $v \in R_i$. Let $w$ be the closest branching vertex to the right of $v$ in $\optsol$. Suppose $w \in R_j$. Let $x$ be the topmost vertex of $\optsol$ in $R_j$. Modify $\optsol$ as follows. Delete all edges of $\optsol$ that leave $R_j$, and add edges from $x$ to all vertices of $\optsol$ in $R_{j+1}$. The resulting tree $Z$ is still an optimal solution. Let $A$ be the restriction of $Z$ to layers $R_i, R_{i+1}, \ldots, R_j$; let $B$ be the restriction of $Z$ to layers $R_{j+1}, R_{j+2},\ldots$. Clearly, 
	$A$ is a cover for $w$ rooted at $v$ and $B$ is a solution for the problem $\sssp(w)$. This completes the proof of the first part. The second part is straightforward.
	\end{proof}

	\begin{figure}
    \begin{center}
    \begin{tikzpicture} [scale=0.88,thick,every node/.style={},every fit/.style={ellipse,draw,inner sep=-12pt,text width=2cm}]
        
    \def \orig {-0.75}
    \def \vgap {0.8}
    \def \hgap {2}
    \def \vtag {0.15}
    \def \htag {0.25}
    
    \foreach \j in {1,2,...,6}
    {
        \begin{scope}[xshift=\j*\hgap*28.45-42.68,yshift=-12,start chain=going below,node distance=8.5mm]
        \foreach \i in {6,7,...,9}
            \node[on chain] (d\j\i) [] {};
        \end{scope}
    }
    
    \node [blue!25,fit=(d16) (d19)] {};
    \node [blue!25,fit=(d26) (d29)] {};
    \node [blue!25,fit=(d36) (d39)] {};
    
    \vertex(vv) at (-1,\orig) [minimum size=4pt, lightgray, fill=lightgray, label=above:$v$] {};

    \foreach \i in {0,...,2}
	    \foreach \j in {0,...,4}
            \vertex(r\i\j) at (\i*\hgap+\hgap-1.5,-\j*\vgap+\orig) [minimum size=4pt, lightgray, fill=lightgray] {};
            
    \draw[red, ultra thin](vv)--(r00)--(r11)--(r21);
    \draw[red, ultra thin](vv)--(r03)--(r12)--(r23);
    \draw[red, ultra thin](vv)--(r04)--(r14);
    
    \vertex at (\hgap-1.5,-3*\vgap+\orig) [minimum size=4pt, blue, fill=blue] {};
    \node at (\hgap-1.5+\htag-0.1,-3*\vgap+\orig+\vtag+0.1) {\footnotesize{$x_1$}};
    
    \vertex at (\hgap-1.5,\orig) [minimum size=4pt, blue, fill=blue] {};
    \node at (\hgap-1.5+\htag,\orig+\vtag) {\footnotesize{$x_2$}};
    
    \vertex at (2*\hgap-1.5,-4*\vgap+\orig) [minimum size=4pt, blue, fill=blue] {};
    \node at (2*\hgap-1.5+\htag,-4*\vgap+\orig+\vtag) {\footnotesize{$x_3$}};
    
    \vertex at (2*\hgap-1.5,-2*\vgap+\orig) [minimum size=4pt, blue, fill=blue] {};
    \node at (2*\hgap-1.5+\htag,-2*\vgap+\orig+\vtag) {\footnotesize{$x_4$}};
    
    \vertex at (3*\hgap-1.5,-3*\vgap+\orig) [minimum size=4pt, blue, fill=blue] {};
    \node at (3*\hgap-1.5+\htag,-3*\vgap+\orig+\vtag) {\footnotesize{$x_5$}};
    
    \vertex at (3*\hgap-1.5,-1*\vgap+\orig) [minimum size=4pt, blue, fill=blue] {};
    \node at (3*\hgap-1.5+\htag,-1*\vgap+\orig+\vtag) {\footnotesize{$x_6$}};
    
    \node at (3*\hgap+0.6,-2*\vgap+\orig) {\Huge{$\rightsquigarrow$}};
    
    
    
    \def \hshi {9.5}
    \def \vout {0.15}
    
    \vertex(vee) at (-2+\hshi,\orig) [minimum size=6pt, green, fill=green, label=above:$v$] {};

    \foreach \i in {3,...,5}
	    \foreach \j in {0,...,4}
            \vertex(r\i\j) at (\hshi+\i*\hgap-6.5,-\j*\vgap+\orig) [minimum size=4pt, lightgray, fill=lightgray] {};
    
    \vertex at (\hshi+3*\hgap-6.5,-3*\vgap+\orig) [minimum size=6pt, white, fill=white] {};
    \vertex(r33in) at (\hshi+3*\hgap-6.5-\vout,-3*\vgap+\orig) [minimum size=4pt, blue, fill=blue] {};
    \vertex(r33ot) at (\hshi+3*\hgap-6.5+\vout,-3*\vgap+\orig) [minimum size=4pt, blue, fill=blue] {};

    \vertex at (\hshi+3*\hgap-6.5,\orig) [minimum size=6pt, white, fill=white] {};    
    \vertex(r30in) at (\hshi+3*\hgap-6.5-\vout,\orig) [minimum size=4pt, blue, fill=blue] {};
    \vertex(r30ot) at (\hshi+3*\hgap-6.5+\vout,\orig) [minimum size=4pt, blue, fill=blue] {};
    
    \vertex at (\hshi+4*\hgap-6.5,-4*\vgap+\orig) [minimum size=6pt, white, fill=white] {};
    \vertex(r44in) at (\hshi+4*\hgap-6.5-\vout,-4*\vgap+\orig) [minimum size=4pt, blue, fill=blue] {};
    \vertex(r44ot) at (\hshi+4*\hgap-6.5+\vout,-4*\vgap+\orig) [minimum size=4pt, blue, fill=blue] {};
    
    \vertex at (\hshi+4*\hgap-6.5,-2*\vgap+\orig) [minimum size=6pt, white, fill=white] {};
    \vertex(r42in) at (\hshi+4*\hgap-6.5-\vout,-2*\vgap+\orig) [minimum size=4pt, blue, fill=blue] {};
    \vertex(r42ot) at (\hshi+4*\hgap-6.5+\vout,-2*\vgap+\orig) [minimum size=4pt, blue, fill=blue] {};
    
    \vertex at (\hshi+5*\hgap-6.5,-3*\vgap+\orig) [minimum size=6pt, white, fill=white] {};
    \vertex(r53in) at (\hshi+5*\hgap-6.5-\vout,-3*\vgap+\orig) [minimum size=4pt, blue, fill=blue] {};
    \vertex(r53ot) at (\hshi+5*\hgap-6.5+\vout,-3*\vgap+\orig) [minimum size=4pt, blue, fill=blue] {};
    
    \vertex at (\hshi+5*\hgap-6.5,-1*\vgap+\orig) [minimum size=6pt, white, fill=white] {};
    \vertex(r51in) at (\hshi+5*\hgap-6.5-\vout,-1*\vgap+\orig) [minimum size=4pt, blue, fill=blue] {};
    \vertex(r51ot) at (\hshi+5*\hgap-6.5+\vout,-1*\vgap+\orig) [minimum size=4pt, blue, fill=blue] {};
    
    \draw[red!50, ultra thick](vee)--(r30in);
    \draw[red!50, ultra thick](vee)--(r33in);
    \draw[red!50, ultra thick](vee)--(r34)--(r44in);
    
    \draw[red!50, ultra thick](r30ot)--(r41)--(r51in);
    \draw[red!50, ultra thick](r33ot)--(r42in);
    \draw[red!50, ultra thick](r42ot)--(r53in);
    
    \vertex(zee) at (\hshi+4.9,-2*\vgap+\orig) [minimum size=6pt, yellow, fill=yellow, label=below:$z$] {};
    
    \draw[yellow, thick](r30in) to[out=30,in=100] (zee);
    \draw[yellow, thick](r51in) to[out=50,in=120] (zee);
    \draw[yellow, thick](r33in) to[out=50,in=155] (zee);
    \draw[yellow, thick](r42in) to[out=30,in=165] (zee);
    \draw[yellow, thick](r53in) to[out=45,in=195] (zee);
    \draw[yellow, thick](r44in) to[out=40,in=250] (zee);
    
    \draw[green, thick](vee) to[out=15,in=135]  (r30ot);
    \draw[green, thick](vee) to[out=30,in=150]  (r51ot);
    \draw[green, thick](vee) to[out=320,in=125] (r33ot);
    \draw[green, thick](vee) to[out=330,in=140] (r42ot);
    \draw[green, thick](vee) to[out=340,in=135] (r53ot);
    \draw[green, thick](vee) to[out=275,in=210] (r44ot);
    
    
    \end{tikzpicture}
    \caption{A cover rooted at $v$ is transformed into a max-flow problem as follows. (i) Each $x_i$ is split into $\xiin$ and $\xiout$. (ii) $v$ is connected to all $\xiout$'s. (iii) All $\xiin$'s are connected to a sink $z$.}
    \label{fig:maxflo}
    \end{center}
    \end{figure}
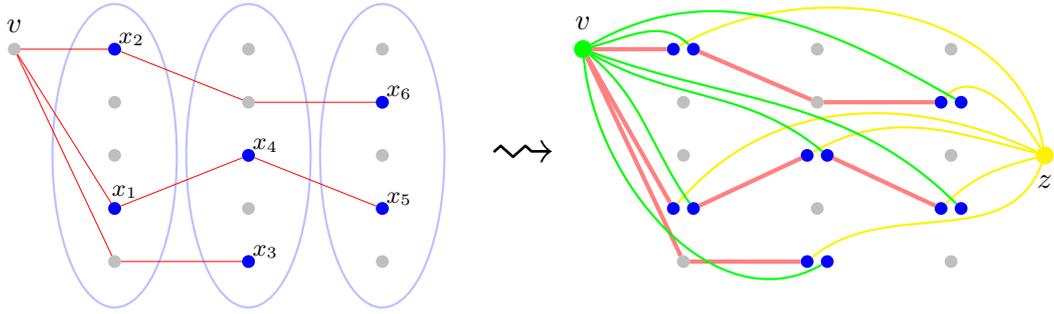

	\subsection{Covers}
	
	In this section, we prove the following lemma.
	\begin{Lemma}
	    There exists a polynomial time algorithm that given $v\in V$ and $\emptyset\neq X\subseteq V$, outputs ``YES'' if $v$ covers $X$, and ``NO'' otherwise. If the answer is ``YES'', the algorithm produces the cover as well.
	\end{Lemma}
	A simple cover for $X$ would be a set of paths from $v$ to vertices in $X$ that have no vertices in common except $v$. To determine if such a cover exists, we declare $v$ to be the source, add a new sink vertex $z$, and edges of the form $(x,z)$ for each $x$ in $X$. Then, well-known algorithms based on network flows in directed graphs can be used to efficiently determine if this network has $|X|$ source to sink (internally) vertex-disjoint paths~\footnote{Text books (see~\cite[Chapter 26]{clrs}) often present only the algorithm for finding edge-disjoint paths; the problem of finding vertex-disjoint paths efficiently reduces to the problem of finding edge-disjoint paths (see~\cite[Theorem 32.4]{vanlintwilson}).}.
	However, a cover need not consist of $|X|$ vertex-disjoint paths; some of the paths might include multiple vertices from $X$. Network flows can still be used to find such covers. In the algorithm $\textsc{RunForCover}()$ presented below, a new sink is added, each $x \in X$ is replaced by two vertices $\xin$ and $\xout$, and the edges incident on $x$ are redirected appropriately. As we argue below, finding a cover is equivalent to finding $|X|$ vertex-disjoint paths in this new graph. 
	\begin{algorithm}
	\begin{algorithmic}[1]
        \State $\Gnew\gets \GBFS$ \Comment{Copy $\GBFS$ into $\Gnew$.}
        \State $V(\Gnew)\gets V(\Gnew)\cup\{z\}$ \Comment{Add a sink vertex $z$ to $\Gnew$.}
	    \For{each $x\in X$}
            \State $V(\Gnew)\gets V(\Gnew)\cup\{\xin,\xout\}$ \Comment{Add two vertices $\xin$ and $\xout$ to $\Gnew$.}
            \State $\inn(\xin)\gets \inn(x)$ \Comment{$\inn(x)\triangleq\{w:(w,x)\in E(\Gnew)\}$.}
            \State $\out(\xout)\gets \out(x)$ \Comment{$\out(x)\triangleq\{w:(x,w)\in E(\Gnew)\}$.}
            \State $\out(\xin)\gets \{z\}, \inn(\xout)\gets \{v\}$
            \State $V(\Gnew)\gets V(\Gnew)-\{x\}$ \Comment{Delete $x$ from $\Gnew$.}
        \EndFor
        \State $F \gets\textsc{MaxVertDisj}(\Gnew,v,z)$ \Comment{$F$ is a set containing the maximum number of vertex-disjoint paths from $v$ to $z$ in $\Gnew$.}
        \If{$|F|<|X|$}
            \State \Return ``NO''
        \Else
            \State \Return ``YES'' \Comment{Translate the paths in $F$ to a cover in $\GBFS$ (details omitted).}
        \EndIf
    \end{algorithmic}\caption{$\textsc{RunForCover}(\GBFS,v,X)$}
    \end{algorithm}
	\myparagraph{Correctness:} See~\autoref{fig:maxflo} for an illustration of $\textsc{RunForCover}()$. Suppose the algorithm outputs ``YES''. We will show that in $\GBFS$ there is a cover for $X$ rooted at $v$. Let $H$ be the union of the paths in $F$ (line 10 of $\textsc{RunForCover}()$). Thus, each vertex in $H$ (except $v$ and $z$) has in-degree and out-degree equal to $1$. For each $x\in X$, add the edge $(\xin,\xout)$ to $E(H)$, and delete all edges \emph{coming in} to $\xout$ and all edges \emph{going out} of $\xin$. Now $z$ is an isolated vertex; delete it from $H$. Note that each vertex of $H$ (except $v$) now has in-degree $1$, and $v$ has in-degree $0$. Thus, $H$ is a tree~\footnote{Every DAG with exactly one vertex of in-degree 0 and all other vertices of in-degree 1 is a tree.}. Furthermore, each vertex of $H$ (except $v$) has out-degree at most $1$ (vertices of the form $\xout$ being the only ones with out-degree $0$). Thus, after contracting edges of the form $(\xin,\xout)$ to a single vertex $x$, the tree $H$ satisfies all the properties of a cover (see~\autoref{coverit}), and all the vertices and edges of $H$ are in $\Gright$.
	
	Now suppose $\fX$ is a cover in $\GBFS$ for $X$ rooted at $v$. We will show that there are $|X|$ vertex-disjoint paths from $v$ to $z$ in $\Gnew$. Fix an $x\in X-\{v\}$. Consider the unique $v\rightarrow x$ path; let $\xprev$ be the vertex in $\{v\} \cup (X-\{x\})$ that appears last on this path. That is, $\xprev$ and $x$ are successive vertices of $\{v\}\cup X$ on the unique path from $v$ to $x$ (note that $\xprev$ can be $v$ itself). Then, the sub-path from $\xprev$ to $x$ in $\fX$ contains no other vertices of $\{v\}\cup X$. In $\Gnew$, either $\xprev=v$ or there is an edge from $v$ to $\xprev_{\operatorname{out}}$. Also, $\Gnew$ has an edge from $\xin$ to $z$. Thus, the path $v\rightarrow \xprev_{\operatorname{out}} \rightarrow \xin \rightarrow z$ (or $v \rightarrow \xin \rightarrow z$, in case $\xprev=v$) is a path from $v$ to $z$; this path is, by construction, unique to $x$, and vertex-disjoint from all other such paths. We thus obtain $|X|$ vertex disjoint $v \rightarrow z$ paths in $\Gnew$.
	
	\myparagraph{Efficiency:} The algorithm $\textsc{MaxVertDisj}()$ can be implemented in polynomial time using network flows~\cite[Theorem 32.4]{vanlintwilson}. Every other step of the algorithm $\textsc{RunForCover}()$ runs in polynomial time.
	
	\subsection{Generalizing to all Interval Graphs}
	
	Now that we have an efficient algorithm to determine the optimal solution for the single-source shortest path problem in interval graphs assuming that the source was the leftmost interval (see~\autoref{main1again}), we can solve the problem for interval graph with no restrictions on the position of the source. The ideas, though, are similar to the restricted case.
	
	\myparagraph{BFS:} We start by performing a breadth-first search starting at the source $s$. As before, this partitions the vertices into layers. We discard edges with both end points in the same layer, and direct the remaining edges away from $s$. Let the layers be $\{s\}=U_0, U_1, U_2, \ldots, U_\alpha$. As the BFS proceeds, we encounter intervals that are at various distances from $s$. Intervals that are at distance two or more from $s$ lie entirely to the left or entirely to the right of $s$, and these vertices cannot be adjacent. Thus, for $i\geq 2$, we can partition $U_i$ as $L_i \cup R_i$, such that all edges leaving $L_i$ enter $L_{i+1}$ and all edges leaving $R_i$ enter $R_{i+1}$. Let $L[i,j] = L_i \cup L_{i+1} \cup \cdots \cup L_j$; similarly, let $R[i,j] = R_i \cup R_{i+1} \cup \cdots \cup R_j$ and $U[i,j] = U_i \cup U_{i+1} \cup \cdots \cup U_j$. For a vertex $u \in L_j$, we consider the problem $\sssp^L(u)$ of constructing a directed tree rooted at $u$ that provides shortest paths to all terminals in $L[j+1,\alpha]$. Similarly, for $u\in R_j$, we define $\sssp^R(u)$ as the problem of constructing a directed tree rooted at $u$ that provides shortest paths to all terminals in $R[j+1,\alpha]$ (some of the later layers may be empty). Let $D^L[u]$ be the minimum number of branching vertices in a solution for $\sssp^L(u)$; let $D_1^L[u]$ be the minimum number of branching vertices in a solution for $\sssp^L(u)$ where $u$ has out-degree $1$. Let $D^R[u]$ and $D^R_1[u]$ be the corresponding quantities for $\sssp^R(u)$. For $u \in U_1$, let $\sssp(u)$ represent the problem of constructing a tree rooted at $u$ that provides shortest paths to terminals in $U[2,\alpha]$; let $D_1[u]$ and $D[u]$ be the corresponding values.
	
	\myparagraph{Dynamic programming:} One can compute $D_1^L[u], D^L[u], D_1^R[u]$ and $D[u]$ using the method presented in the algorithm $\textsc{FindOpt}()$ \emph{mutatis mutandis}. We now focus on vertices $u \in U_1$ (we will get to $s$ eventually). Now, $D_1[u] = D_1^R[u]$ if $L[2,\alpha] \cap T = \emptyset$ and $D_1[u] = D_1^L[u]$ if $R[2,\alpha] \cap T = \emptyset$; otherwise,  $D_1[v]=\infty$.  Similarly, we can compute $D[u]$ when either $L[2,\alpha] \cap T = \emptyset$ or $R[2,\alpha] \cap T = \emptyset$. If both $L[2,\alpha] \cap T \neq \emptyset$ and $R[2,\alpha]\cap T \neq \emptyset$, then we will rely on a decomposition of the optimal solution. First, we need a definition. Let $w_L \in L_i$ and $w_R \in L_j$.
	We say that $\fX$ is a \emph{cover} for $(w_L, w_R)$, rooted at $u$ if there exist vertex-disjoint paths from $u$ to $((L[2,i] \cap T) \cup (L[2,j]\cap T))\cup \{w_L, w_R\}$.
	\begin{Lemma}
	Suppose $L[2,\alpha] \cap T \neq \emptyset$ and $R[2,\alpha]\cap  T\neq \emptyset$. Let $u \in U_1$. Then, there are vertices $w_L \in L[2,\alpha]$ and $w_R \in R[2,\alpha]$, and an optimal solution $Z$ for $\sssp(u)$ of the form $Z=A \cup B_L \cup B_R$, such that $A$ is a cover for $(w_L,w_R)$ rooted at $u$, $B_L$ is a solution for $\sssp^L(w_L)$ and $B_R$ is a solution for 
	$\sssp^R(w_R)$. (We also need the converse for our proof, but we do not state it here explicitly.)
	\end{Lemma}
	(Proof omitted.)
	Once we have such a solution, we can set the value of $D[u]$ to be the minimum of $D_1[u]$ and the minimum value taken by $D[w_L] + D[w_R] +1$, as $(w_L,w_R)$ ranges over $L[2,\alpha] \times R[2,\alpha]$.
	
	Finally, we consider the problem $\sssp(s)$ of constructing a tree rooted at the source $s$ that provides shortest paths to all terminals; let $D_1[s]$ be the minimum number of branching vertices in a solution where $s$ has out-degree $1$; let $D[s]$ the number of branching vertices in a solution. For computing $D_1[s]$ we proceed along expected lines, with different cases based on how many terminals there are in $U_1$. For $D[s]$, we again rely on a decomposition.
	\begin{Lemma}
	There is an optimal solution of $\sssp(s)$ of one of the following forms.
	\begin{enumerate}
	    \item[(a)] There is a vertex $u \in U_1$ and the solution consists of the edge $(s,u)$ together with a solution $B$ to $\sssp(u)$;
	    \item[(b)] There is a pair of vertices $(w_L,w_R) \in U[1,\alpha] \times U[1,\alpha]$
	    and the solution has the form $Z=A \cup B_L \cup B_R$, where 
	    $A$ is a cover for $(w_L,w_R)$ rooted at $s$, $B_L$ is a solution for $\sssp^L(w_L)$ and $B_R$ is a solution for $\sssp^R(w_R)$;
	    \item[(c)] There is a vertex $u \in U_1$, a pair of vertices $(w_L,w_R) \in L[2,\alpha] \times[(c)] R[2,\alpha]$ and a solution of the form $Z = \{(s,u)\} \cup A \cup B_L \cup B_R$, where $B_L$ is a solution to $\sssp^L(w_L)$, $B_R$ is a solution to $\sssp^R(w_R)$, and $A$ is a cover for $(w_L,w_R)$ rooted at the the pair
	    $(s,u)$ (that is, $A$ consists of vertex disjoint paths, each originating at either $s$ or $u$, that include all terminals in the layers up to $w_L$ and $w_R$ and
	    also include $w_L$ and $w_R$.
	\end{enumerate}
	We also have the appropriate converse statement in each of the three cases.
	\end{Lemma}
	With this, we can compute $D[s]$. For case (a), we have already computed $D_1[s]$; for case (b), we take the minimum of $D[w_L] + D[w_R] + 1$ as we range over all feasible choices of $(w_L,w_R)$; for case (c), we take the minimum of $D[w_L] + D[w_R] + 2$ as we range over all feasible choices of $(u,w_L,w_R)$. This leads to a polynomial time dynamic programming solution for $\sssp(s)$.

	
	

    \section{All-pairs Distance-preserving Subgraphs}

    Interval graphs are ``one-dimensional" by nature. It is therefore interesting to consider two-dimensional analogues of interval graphs. There are several different ways to generalize interval graphs into two (or more) dimensions, and many of them have been studied in literature.
    
    The main theorem for this section is the following, which is simply~\autoref{main2} restated.
    
    \begin{Theorem} [All-pairs bi-interval graphs] \ 
	    \begin{enumerate}
			\item[(a)] There is a polynomial time algorithm that, given a bi-interval graph with $k$ terminals as input, produces a distance-preserving subgraph of it with $O(k^2)$ branching vertices.
			\item[(b)] For every $k\geq 4$, there is a bi-interval graph $\Gdiag$ on $k$ terminals such that every distance-preserving subgraph of $\Gdiag$ has $\Omega(k^2)$ branching vertices.
		\end{enumerate}
	\end{Theorem}
	
	We shall see a proof of this shortly. First, let us examine some other 2D analogues of interval graphs.

    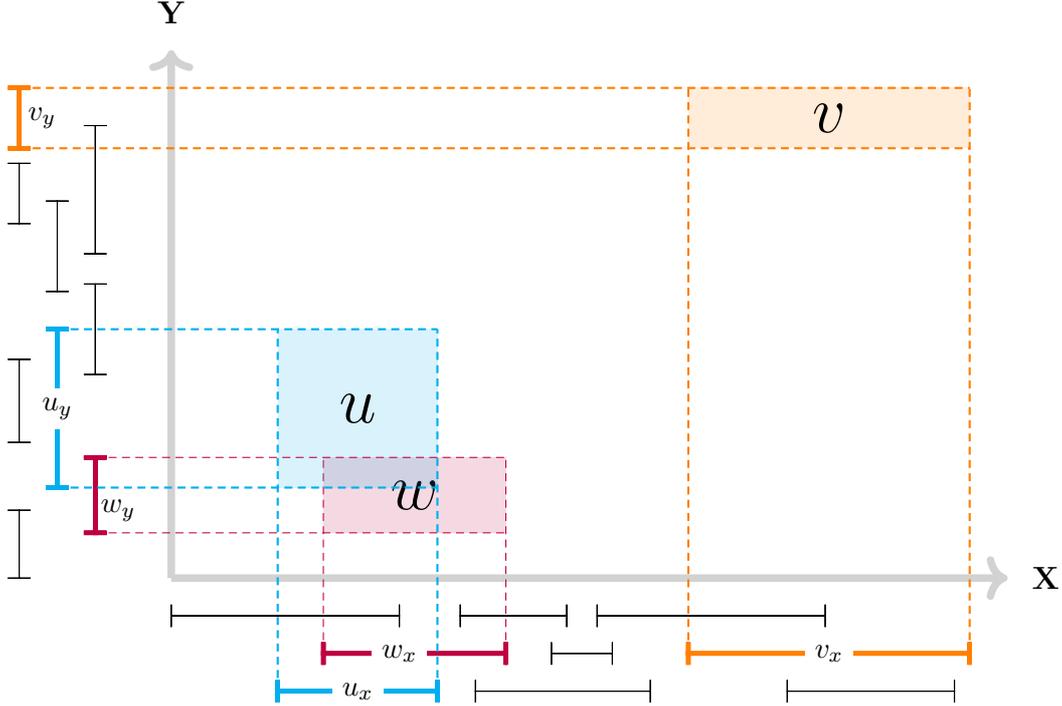
\begin{figure}
    \begin{center}
    \begin{tikzpicture}
    
    \draw [->, gray!35, line width=1mm] (0,0) to (11,0);
    \node at (11.5,0) {\large{\bf X}};
    \draw [->, gray!35, line width=1mm] (0,0) to (0,7);
    \node at (0,7.5) {\large{\bf Y}};
    
    \def \d {0.15};
    \def \colouru {cyan};
	\def \colourv {orange};
	\def \colourw {purple};

	
	\def \x      {{ 0  , 2  , 3.8, 4  , 5  , 5.6, 8.1, 6.8, 1.4}};
	\def \y      {{-0.5,-1  ,-0.5,-1.5,-1  ,-0.5,-1.5,-1  ,-1.5}};
	\def \length {{ 3  , 2.4, 1.4, 2.3, 0.8, 3  , 2.2, 3.7, 2.1}};
	
	\foreach \i in {0,1,...,6}
	    \draw [line width=0.1mm] (\x[\i],\y[\i]+\d)--(\x[\i],\y[\i]-\d)--(\x[\i],\y[\i])--(\x[\i]+\length[\i],\y[\i])--(\x[\i]+\length[\i],\y[\i]-\d)--(\x[\i]+\length[\i],\y[\i]+\d)--(\x[\i]+\length[\i],\y[\i])--(\x[\i],\y[\i])--(\x[\i],\y[\i]+\d);

    \def \i {1};
    \draw [line width=0.6mm,\colourw] (\x[\i],\y[\i]+\d)--(\x[\i],\y[\i]-\d)--(\x[\i],\y[\i])--(\x[\i]+\length[\i],\y[\i])--(\x[\i]+\length[\i],\y[\i]-\d)--(\x[\i]+\length[\i],\y[\i]+\d)--(\x[\i]+\length[\i],\y[\i])--(\x[\i],\y[\i])--(\x[\i],\y[\i]+\d);
    
    \draw [densely dashed,thin,\colourw] (\x[\i],\y[\i])--(\x[\i],1.6);
    \draw [densely dashed,thin,\colourw] (\x[\i]+\length[\i],\y[\i])--(\x[\i]+\length[\i],1.6);
    
    \fill[fill=\colourw, opacity=0.15](\x[\i],0.6) rectangle +(2.4,1);
    \node at (\x[\i]+1.2,1.1) {\Huge{$w$}};
    
    \def \i {7};
    \draw [line width=0.6mm,\colourv] (\x[\i],\y[\i]+\d)--(\x[\i],\y[\i]-\d)--(\x[\i],\y[\i])--(\x[\i]+\length[\i],\y[\i])--(\x[\i]+\length[\i],\y[\i]-\d)--(\x[\i]+\length[\i],\y[\i]+\d)--(\x[\i]+\length[\i],\y[\i])--(\x[\i],\y[\i])--(\x[\i],\y[\i]+\d);
    
    \draw [densely dashed,thick,\colourv] (\x[\i],\y[\i])--(\x[\i],6.5);
    \draw [densely dashed,thick,\colourv] (\x[\i]+\length[\i],\y[\i])--(\x[\i]+\length[\i],6.5);
    
    \fill[fill=\colourv, opacity=0.15](\x[\i],5.7) rectangle +(3.7,0.8);
    \node at (\x[\i]+1.85,6.1) {\Huge{$v$}};
    
    \def \i {8};
    \draw [line width=0.6mm,\colouru] (\x[\i],\y[\i]+\d)--(\x[\i],\y[\i]-\d)--(\x[\i],\y[\i])--(\x[\i]+\length[\i],\y[\i])--(\x[\i]+\length[\i],\y[\i]-\d)--(\x[\i]+\length[\i],\y[\i]+\d)--(\x[\i]+\length[\i],\y[\i])--(\x[\i],\y[\i])--(\x[\i],\y[\i]+\d);
	
	\draw [densely dashed,thick,\colouru] (\x[\i],\y[\i])--(\x[\i],3.3);
    \draw [densely dashed,thick,\colouru] (\x[\i]+\length[\i],\y[\i])--(\x[\i]+\length[\i],3.3);
	
	\fill[fill=\colouru, opacity=0.15](\x[\i],1.2) rectangle +(2.1,2.1);
	\node at (\x[\i]+1.05,2.25) {\Huge{$u$}};
	
	
	\def \y      {{ 0  , 0.6, 1.8, 2.7, 3.8, 4.3, 4.7, 5.7, 1.2}};
	\def \x      {{-2  ,-1  ,-2  ,-1  ,-1.5,-1  ,-2  ,-2  ,-1.5}};
	\def \length {{ 0.9, 1  , 1.1, 1.2, 1.2, 1.7, 0.8, 0.8, 2.1}};
	
    \foreach \i in {0,1,...,6}
	    \draw [line width=0.1mm] (\x[\i]+\d,\y[\i])--(\x[\i]-\d,\y[\i])--(\x[\i],\y[\i])--(\x[\i],\y[\i]+\length[\i])--(\x[\i]-\d,\y[\i]+\length[\i])--(\x[\i]+\d,\y[\i]+\length[\i])--(\x[\i],\y[\i]+\length[\i])--(\x[\i],\y[\i])--(\x[\i]+\d,\y[\i]);

    \def \i {1};
    \draw [line width=0.6mm,\colourw] (\x[\i]+\d,\y[\i])--(\x[\i]-\d,\y[\i])--(\x[\i],\y[\i])--(\x[\i],\y[\i]+\length[\i])--(\x[\i]-\d,\y[\i]+\length[\i])--(\x[\i]+\d,\y[\i]+\length[\i])--(\x[\i],\y[\i]+\length[\i])--(\x[\i],\y[\i])--(\x[\i]+\d,\y[\i]);
    
    \draw [densely dashed,thin,\colourw] (\x[\i],\y[\i])--(4.4,\y[\i]);
    \draw [densely dashed,thin,\colourw] (\x[\i],\y[\i]+\length[\i])--(4.4,\y[\i]+\length[\i]);

    \def \i {7};
    \draw [line width=0.6mm,\colourv] (\x[\i]+\d,\y[\i])--(\x[\i]-\d,\y[\i])--(\x[\i],\y[\i])--(\x[\i],\y[\i]+\length[\i])--(\x[\i]-\d,\y[\i]+\length[\i])--(\x[\i]+\d,\y[\i]+\length[\i])--(\x[\i],\y[\i]+\length[\i])--(\x[\i],\y[\i])--(\x[\i]+\d,\y[\i]);
    
    \draw [densely dashed,thick,\colourv] (\x[\i],\y[\i])--(10.5,\y[\i]);
    \draw [densely dashed,thick,\colourv] (\x[\i],\y[\i]+\length[\i])--(10.5,\y[\i]+\length[\i]);
    
    \def \i {8};
    \draw [line width=0.6mm,\colouru] (\x[\i]+\d,\y[\i])--(\x[\i]-\d,\y[\i])--(\x[\i],\y[\i])--(\x[\i],\y[\i]+\length[\i])--(\x[\i]-\d,\y[\i]+\length[\i])--(\x[\i]+\d,\y[\i]+\length[\i])--(\x[\i],\y[\i]+\length[\i])--(\x[\i],\y[\i])--(\x[\i]+\d,\y[\i]);
    
    \draw [densely dashed,thick,\colouru] (\x[\i],\y[\i])--(3.5,\y[\i]);
    \draw [densely dashed,thick,\colouru] (\x[\i],\y[\i]+\length[\i])--(3.5,\y[\i]+\length[\i]);
    
    \node[fill=white] at (2.45,-1.5) {$u_x$};
    \node[fill=white] at (-1.5,2.25) {$u_y$};
    
    \node[fill=white] at (8.65,-1) {$v_x$};
    \node at (-1.7,6.1) {$v_y$};
    
    \node[fill=white] at (3,-1) {$w_x$};
    \node at (-0.7,0.9) {$w_y$};
    
    
    \end{tikzpicture}
    \caption{Three vertices $u,v,w$, represented as rectangles in a bi-interval graph. Their corresponding intervals in the interval graphs $X$ and $Y$ are $u_x,v_x,w_x$ and $u_y,v_y,w_y$, respectively. Note that $(u,w)$ is an edge in the bi-interval graph since the rectangles $u$ and $w$ overlap, but $(u,v)$ and $(v,w)$ are not.}
    \label{fig:biinter}
    \end{center}
    \end{figure}

    \subsection{Two-dimensional Analogues of Interval Graphs}
    
    A bi-interval graph $G$ is determined by two families of intervals $\mathcal{I}_X$ and $\mathcal{I}_Y$. There is a vertex  $v_{a,b}$ in  $V(G)$ for each pair $(a,b) \in \mathcal{I}_X \times \mathcal{I}_Y$. This criterion used to connect two such vertices by an edge has a natural geometric interpretation:  two vertices are adjacent if the rectangles associated with them intersect. Formally, $(a_1,b_1)$ and $(a_2, b_2)$ are connected by an edge if and only if $(a_1 \times b_1) \cap (a_2 \times b_2) \neq \emptyset$. Thus, bi-interval graphs are rectangle intersection graphs, where the rectangles that appear have a certain product structure induced by two families of intervals (see~\autoref{fig:biinter}).
    
    Bi-interval graphs are a natural generalization of interval graphs in the sense that they have boxicity two~\footnote{The boxicity of a graph is the minimum dimension in which a given graph can be represented as an intersection graph of axis-parallel boxes.}, and interval graphs by definition are precisely the graphs having boxicity one. Thus, it might be interesting to study the class of graphs having boxicity two (or higher). However, the following theorem shows that not much can be done for such graphs.
    
    \begin{Theorem} \label{boxicity}
    For every $k\geq 4$, there exists a graph having boxicity two on $k$ terminals for which every distance-preserving subgraph has $\Omega(k^4)$ branching vertices.
    \end{Theorem}
     
    \begin{proof}[Proof sketch] It is easy to see that the weighted planar graph $G$ presented in~\cite[Section 5]{Ngu} can be made unweighted (by subdividing the edges) so that every distance-preserving subgraph of $G$ has $\Omega(k^4)$ branching vertices. Now that $G$ is unweighted and planar, it is possible to replace its vertices by axis-parallel rectangles, so that $G$ now has boxicity two.
    \end{proof}
    
    In this paper, we study the class of unweighted bi-interval graphs. It is also reasonable to consider bi-interval graphs with non-negative real edge weights. Using earlier work by Krauthgamer \emph{et al.}~\cite{Ngu}, it can be shown that for every $k\geq 4$, there exists a \emph{weighted} bi-interval graph~\footnote{Every interval graph is a bi-interval graph.} on $k$ terminals for which every distance-preserving subgraph has $\Omega(k^4)$ branching vertices (see~\cite[Corollary 8 (b)]{kshitij}).
	
	None of these lower bounds can be improved apart from constant factors, since an $O(k^4)$ upper bound exists for both weighted and unweighted graphs (see~\cite[Section 2.1]{Ngu}).

	\subsection{Proof of the Upper Bound for Bi-interval Graphs: Idea}
	
	In this section, we prove~\autoref{main2} (a). The proof is constructive by nature. Here we give a rough outline of the proof.
	
	\begin{Theorem} \label{main2aagain}
	There exists a polynomial time algorithm that, given a bi-interval graph $G$ with $k$ terminals, constructs a distance-preserving subgraph of $G$ with $O(k^2)$ branching vertices.
	\end{Theorem}
	
	In order to describe our construction of distance-preserving subgraphs of bi-interval graphs, it will be helpful to fix a method of constructing shortest paths in such graphs. \autoref{fig:pseudop} is helpful to navigate through this construction. It is well known and easy to prove that a greedy algorithm works for finding shortest paths in interval graphs. In a nutshell, we run this greedy algorithm on two interval graphs and put the resulting paths together.
	
	Suppose $G=X\boxtimes Y$ is a bi-interval graph (see~\autoref{def:biinterval}), and we need to construct a shortest path between the vertices $u=(u_x,u_y)$ and $v=(v_x,v_y)$ in $G$ (assume that $u_y\preceq v_y$). $\greedyX(u_x,v_x)=(u_x=i_0,i_1,i_2,\ldots,i_p=v_x)$ and $\greedyY(u_y,v_y)=(u_y=j_0,j_1,j_2,\ldots,j_q=v_y)$ be the greedy shortest paths from $u_x$ to $v_x$ in $X$ and from $u_y$ to $v_y$ in $Y$, respectively. Thus, $d_X(u_x,v_x)=p$ and $d_Y(u_y,v_y)=q$. Suppose $p \leq q$ (otherwise, exchange $X$ and $Y$). Then, $\greedyG(u,v) = ((i_0,j_0),(i_1,j_1),(i_2,j_2),\ldots,(i_p,j_p),(i_p,j_{p+1}),(i_p,j_{p+2}),\ldots,(i_p,j_{q-1}),(i_p,j_q))$ is a shortest path between $u$ and $v$ in $G$. This method can be visualized geometrically: starting from the rectangle $u$, embark on the greedy shortest paths in both $X$ and $Y$ simultaneously until the current rectangle lies in the same row or column as the destination rectangle (that is, move \emph{diagonally}); then, move optimally within the row or column in the corresponding interval graph~\footnote{Note that when restricted to a fixed row (or a fixed column), a bi-interval graph is simply an interval graph.} to reach the destination (\emph{horizontally} or \emph{vertically up}; we call such paths \emph{straight} paths). Since we assumed that $u_y \leq v_y$, the diagonal segments of these paths either go northeast or northwest (but never southeast or southwest~\footnote{We do not consider southeast or northwest paths because southeast is anti-parallel to northwest, and southwest is anti-parallel to northeast, and there does not seem to be any straightforward method to get a handle on the number of branching vertices in terms of $k$ by using anti-parallel paths (see~\autoref{antiparallel}).}).
	
	Our subgraph will consist of such paths for all pairs of terminals. First, we make provisions for the diagonal parts of the shortest paths, by independently moving northeast and northwest along greedy paths. Two paths, both proceeding northeast (or both proceeding northwest) from different terminals can meet, but once they have met they move in unison, never to diverge again. Furthermore, a path proceeding northeast and another proceeding northwest meet at most once and diverge immediately, never to meet again. Thus, by introducing at most $O(k^2)$ branching vertices, we succeed in making provision for all diagonal segments of shortest paths out of terminals.
	
	Two tasks still remain: (i) we must identify vertices at which these diagonal segments branch off into a vertical or a horizontal segment (potentially introducing a branching vertex) and arrive at a vertex (which we call \emph{pseudo-terminal}) in the row or column of another terminal; (ii) ensure that every pseudo-terminal is connected to every terminal in the row or column (with appropriate straight paths). The first task is straightforward for there are no choices to be made. The second task requires some care; we might have multiple terminals in the same column, which need shortest paths to all the pseudo-terminals in that column. To keep the number of branching vertices small, we borrow some ideas from the earlier analyses of distance-preserving subgraphs in interval graphs~\cite{jaikumar}.
	
	The total count of branching vertices breaks down as follows: (i) $O(k^2)$ when diagonal paths are added; (ii) $O(k^2)$ when diagonal paths branch off to join a straight path. Let us explain (ii) briefly. There are at most $O(k)$ pseudo-terminals in any row or column. We argue that within a row (or column) with $p$ terminals and $q$ pseudo-terminals, shortest paths can be provided by introducing $O(p(p+q))$ additional branching vertices. Since $p+q = O(k)$ and the total number of terminals is $k$, we need at most $O(k^2)$ branching vertices for (ii).
	
	We now describe this proof in detail.
	
	\subsection{Proof of the Upper Bound for Bi-interval Graphs: Implementation}

	In this section, we prove~\autoref{main2aagain}. Let us begin by formally defining bi-interval graphs. For this, we need the notion of strong graph products.
	
	\begin{Definition}[Strong product]
    Given two graphs $G_1$ and $G_2$, the strong graph product of $G_1$ and $G_2$, denoted by $G_1\boxtimes G_2$, is defined as follows.

    \begin{itemize}
        \item $V(G_1\boxtimes G_2)=V(G_1)\times V(G_2)$.
        \item Let $(u_1,u_2)$ and $(v_1,v_2)$ be two vertices of $G_1\boxtimes G_2$ such that $(u_1,u_2)\neq (v_1,v_2)$. Then $((u_1,u_2),(v_1,v_2))\in E(G_1\boxtimes G_2)$ if and only if one of the following is true.
        \begin{enumerate}
            \item $u_1=v_1$ in $G_1$ and $(u_2,v_2)\in E(G_2)$.
            \item $u_2=v_2$ in $G_2$ and $(u_1,v_1)\in E(G_1)$.
            \item $(u_1,v_1)\in E(G_1)$ and $(u_2,v_2)\in E(G_2)$.
        \end{enumerate}
    \end{itemize}
    In other words, two distinct vertices $(u_1,u_2)$ and $(v_1,v_2)$ are adjacent in $G_1\boxtimes G_2$ if and only if they are adjacent or equal in each coordinate.
    \end{Definition}
    
    We are now set to define bi-interval graphs. If $G_1$ and $G_2$ are interval graphs, then $G_1\boxtimes G_2$ is a bi-interval graph.
    
    \begin{Definition} [Bi-interval graph] \label{def:biinterval}
        A graph $G$ is a bi-interval graph if it can be expressed as the strong product of two interval graphs $G_1$ and $G_2$. This is denoted as $G=G_1\boxtimes G_2$.
    \end{Definition}
    
    We now define a sub-class of bi-interval graphs known as king's graphs. King's graphs will be used in our proof of the lower bound (see~\autoref{main2bagain}). Note that a path graph is an interval graph. When $X$ and $Y$ are both path graphs, then $X\boxtimes Y$ resembles a chess board (see~\autoref{fig:chess}). The vertices are the squares of the chess board, and there is an edge between two vertices of the graph if and only if a king can go from one square to the other in a single move. Such a bi-interval graph is therefore called a king's graph.

    \begin{Definition} [King's graph] \label{king}
        A graph $G$ is a king's graph if $G=G_1\boxtimes G_2$, where $G_1$ and $G_2$ are path graphs.
    \end{Definition}

	We now describe a greedy algorithm to find shortest paths in bi-interval graphs which is similar in flavour to the greedy algorithm to find shortest paths in interval graphs. Let $G=X\boxtimes Y$ be a bi-interval graph, where $X$ is an interval graph on the X-axis, and $Y$ is an interval graph on the Y-axis. Assume for simplicity that both $X$ and $Y$ have $n$ vertices each. Let $V(X)=V(Y)=[n]$ such that $i\prec {i+1}$ for each $1\leq i<n$ in both $\fI_X$ and $\fI_Y$. Thus, $X\boxtimes Y$ has $n^2$ vertices. Let $T\subseteq [n]\times[n]$ denote the set of terminals. (Since $S=T$, we only use $T$ to denote the set of terminals, and $k=|T|$.)
	
	Suppose we need to construct a shortest path between the vertices $u=(u_x,u_y)$ and $v=(v_x,v_y)$ in $G$. We may assume that $u_x\preceq v_x$ and $u_y\preceq v_y$ (the other cases are similar). Let $\greedyX(u_x,v_x)=(u_x=i_0,i_1,i_2,\ldots,i_p=v_x)$ and $\greedyY(u_y,v_y)=(u_y=j_0,j_1,j_2,\ldots,j_q=v_y)$ be the greedy shortest paths from $u_x$ to $v_x$ in $X$ and from $u_y$ to $v_y$ in $Y$, respectively. Thus, $d_X(u_x,v_x)=p$ and $d_Y(u_y,v_y)=q$. Depending on the values of $p$ and $q$, the path $\greedyG(u,v)$ is defined as follows.
	\begin{align*}
	&\text{If }p\leq q\text{, then }\greedyG(u,v)=\greedyG((u_x,u_y),(v_x,v_y))=\greedyG((i_0,j_0),(i_p,j_q))\\
	&=((i_0,j_0),(i_1,j_1),(i_2,j_2),\ldots,(i_p,j_p),(i_p,j_{p+1}),(i_p,j_{p+2}),\ldots,(i_p,j_{q-1}),(i_p,j_q)).\\
	\\
	&\text{If }p>q\text{, then }\greedyG(u,v)=\greedyG((u_x,u_y),(v_x,v_y))=\greedyG((i_0,j_0),(i_p,j_q))\\
	&=((i_0,j_0),(i_1,j_1),(i_2,j_2),\ldots,(i_q,j_q),(i_{q+1},j_q),(i_{q+2},j_q),\ldots,(i_{p-1},j_q),(i_p,j_q)).
	\end{align*}
	Another way to view $\greedyG(u,v)$ is in the geometric setting. Starting from the rectangle $u$, embark on the greedy shortest paths in both $X$ and $Y$ \emph{simultaneously}, until the current rectangle lies in the same row or column as the destination rectangle, in which case the problem boils down to that of a greedy shortest path in an interval graph.
	
	\begin{Claim}
	The path $\greedyG(u,v)$ is a shortest path from $u$ to $v$ in $G$.
	\end{Claim}
	
	\begin{proof}
	Note that any path $P_G(u,v)$ in $G$ induces a path $P_X(u_x,v_x)$ in $X$ and a path $P_Y(u_y,v_y)$ in $Y$. Since $P_X(u_x,v_x)\geq\greedyX(u_x,v_x)$ and $P_Y(u_y,v_y)\geq\greedyY(u_y,v_y)$, we have $d_G(u,v)\geq \max\{d_X(u_x,v_x),d_Y(u_y,v_y)\}=\max\{p,q\}$. Since $\greedyG(u,v)$ always produces a path of length $\max\{p,q\}$, this completes the proof.
	\end{proof}

	\subsection{Constructing the Subgraph}
	
	Given a bi-interval graph $G=X\boxtimes Y$, we now describe the construction of $H$, a distance-preserving subgraph of $G$ with $O(k^2)$ branching vertices, thereby proving~\autoref{main2} (a). The idea behind this construction is as follows. $H$ is initially the empty graph on the vertex set of $G$. We add edges to $H$ in a systematic manner, using two algorithms $\textsc{PseudoTerms}()$ and $\textsc{StraightPaths}()$.
	
    We execute these algorithms sequentially. Firstly, $\textsc{PseudoTerms}()$ converts certain non-terminal vertices of the graph to ``pseudo-terminal'' vertices, effectively reducing the problem from one bi-interval graph to a problem on several interval graphs. Secondly, $\textsc{StraightPaths}()$ computes shortest paths within those interval graphs. These algorithms require the concept of cardinal paths. In simple terms, a cardinal path is an ``unbounded'' extension of a greedy path.
	
	\begin{Definition}[Cardinal paths]
	Cardinal paths are greedy paths with a fixed source but no fixed destination. See~\autoref{fig:pseudop} for examples.
	\begin{itemize}
	    \item Given an interval graph $X$ and a vertex $i_0\in V(X)$, the path $\east(i_0,\infty)=(i_0,i_1,i_2,\ldots)$, called the $i_0$-\texttt{east} path, is the greedy path originating at $i_0$ and travelling eastward with no fixed destination. In other words, the path starts with the interval $i_0$ and intervals are greedily added to it until there are no more intervals to add. (That is, until the interval $i$ with the maximum $\righty(i)$ is reached. In our case, $i=n$.) Paths like $\texttt{north}$, $\texttt{east}$, $\texttt{west}$, etc. are similarly defined.
	    \item Given a bi-interval graph $G=X\boxtimes Y$ and a vertex $u=(i_0,j_0)\in V(G)$, the $u$-\textit{\texttt{northeast}} path, denoted by $\neast(u,\infty)$, is defined as $$\neast(u,\infty)=((i_0,j_0),(i_1,j_1),(i_2,j_2),\ldots).$$
	    Here, $\east(i_0,\infty)=(i_0,i_1,i_2,\ldots)$ is the $i_0$-\texttt{east} path, and $\north(j_0,\infty)=(j_0,j_1,j_2,\ldots)$ is the $j_0$-\texttt{north} path. ($\neast$ continues to add vertices to its path as long as both $\north$ and $\east$ continue to add vertices to their respective paths.) The path $\textit{\texttt{northwest}}$ is similarly defined.
	\end{itemize}
	For the rest of the proof, we will only be using $\textit{\texttt{northeast}}$ and $\textit{\texttt{northwest}}$ paths.
	\end{Definition}

	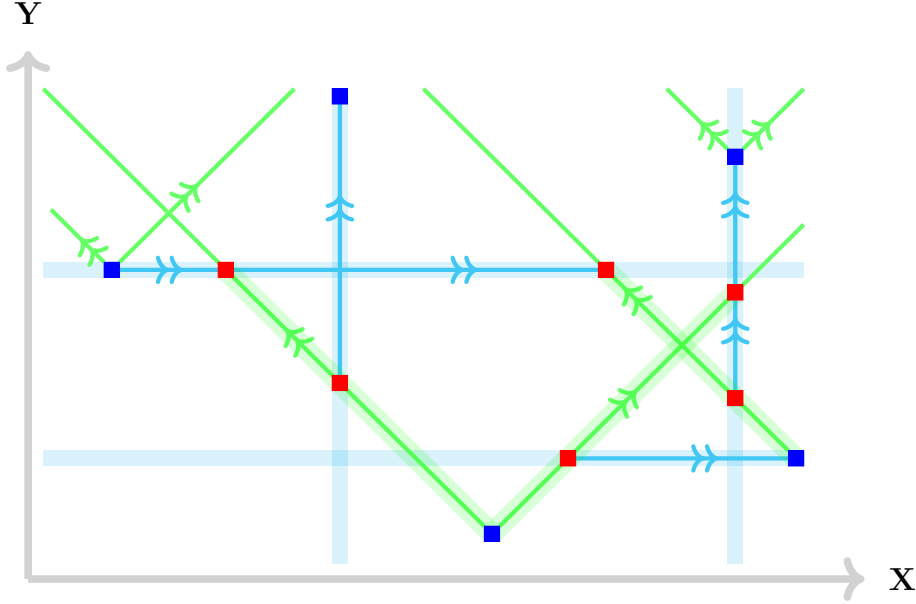
\begin{figure}
    \begin{center}
    \begin{tikzpicture}
    
    \draw [->, gray!35, line width=1mm] (0,0.5) to (11,0.5);
    \node at (11.5,0.5) {\large{\bf X}};
    \draw [->, gray!35, line width=1mm] (0,0.5) to (0,7.5);
    \node at (0,8) {\large{\bf Y}};
    
    \fill[fill=cyan, opacity=0.15] (0.2,2) rectangle (10.2,2.2);
    \fill[fill=cyan, opacity=0.15] (0.2,4.5) rectangle (10.2,4.7);
    \fill[fill=cyan, opacity=0.15] (4,0.7) rectangle (4.2,7);
    \fill[fill=cyan, opacity=0.15] (9.2,0.7) rectangle (9.4,7);
    
    \begin{scope} [green!60, ultra thick, decoration={markings, mark=at position 0.45 with {\arrow{>>}}}] 
        \draw[postaction={decorate}] (6,1.2)--(0.2,7);
        \draw[postaction={decorate}] (6.2,1.2)--(10.2,5.2);
        \draw[postaction={decorate}] (1,4.7)--(0.3,5.4);
        \draw[postaction={decorate}] (1.2,4.7)--(3.5,7);
        \draw[postaction={decorate}] (10,2.2)--(5.2,7);
        \draw[postaction={decorate}] (9.2,6.2)--(8.4,7);
        \draw[postaction={decorate}] (9.4,6.2)--(10.2,7);
    \end{scope}
    
    \begin{scope} [cyan!60, ultra thick, decoration={markings, mark=at position 0.65 with {\arrow{>>}}}] 
        \draw[postaction={decorate}] (1,4.6)--(2.6,4.6);
        \draw[postaction={decorate}] (2.6,4.6)--(7.7,4.6);
        \draw[postaction={decorate}] (7,2.1)--(10.2,2.1);
        \draw[postaction={decorate}] (4.1,3)--(4.1,7);
        \draw[postaction={decorate}] (9.3,2.8)--(9.3,4.6);
        \draw[postaction={decorate}] (9.3,4.6)--(9.3,6.2);
    \end{scope}
    
    \fill [fill=green, opacity=0.15] (6,1)--(6.2,1.2)--(2.7,4.7)--(2.5,4.5)--cycle;
    \fill [fill=green, opacity=0.15] (6.2,1)--(6,1.2)--(9.2,4.4)--(9.4,4.2)--cycle;
    \fill [fill=green, opacity=0.15] (10,2)--(10.2,2.2)--(7.7,4.7)--(7.5,4.5)--cycle;
    
    \draw [blue,fill=blue] (6,1) rectangle (6.2,1.2);
    \draw [red ,fill=red ] (4,3) rectangle (4.2,3.2);
    \draw [blue,fill=blue] (4,6.8) rectangle (4.2,7);
    \draw [red ,fill=red ] (9.2,4.2) rectangle (9.4,4.4);
    \draw [blue,fill=blue] (9.2,6) rectangle (9.4,6.2);
    \draw [red ,fill=red ] (2.5,4.5) rectangle (2.7,4.7);
    \draw [blue,fill=blue] (1,4.5) rectangle (1.2,4.7);
    \draw [red ,fill=red ] (7,2) rectangle (7.2,2.2);
    \draw [blue,fill=blue] (10,2) rectangle (10.2,2.2);
    \draw [red ,fill=red ] (9.2,2.8) rectangle (9.4,3);
    \draw [red ,fill=red ] (7.5,4.5) rectangle (7.7,4.7);
    
    
    \end{tikzpicture}
    \caption{The structure of a distance-preserving subgraph of a bi-interval graph. The blue vertices are the terminals, the green paths denote cardinal \textit{\texttt{northwest}} and \textit{\texttt{northeast}} paths originating at the terminals, the blue paths denote straight (horizontal/vertical) paths, and the red vertices are the pseudo-terminals that lie on the intersection of green and blue paths.}
    \label{fig:pseudop}
    \end{center}
    \end{figure}

	The rationale behind using \textit{\texttt{northeast}} and \textit{\texttt{northwest}} paths, but not using \textit{\texttt{southeast}} or \textit{\texttt{southwest}} paths is because they might lead to anti-parallel paths (see~\autoref{antiparallel}). We are now set to present the algorithm $\textsc{PseudoTerms}()$. The input to this algorithm is a bi-interval graph $G=X\boxtimes Y$ with a set $T$ of terminals, and its output is a subgraph $H$ of $G$, and a set $\fP$ of pseudo-terminals.
	
	\begin{algorithm}
	\begin{algorithmic}[1]
	    \State $V(H) \gets V(G), E(H) \gets \emptyset$
	    \For{each $v\in T(G)$}
	        \State $E(H) \gets E(H)\cup \neast(v,\infty)\cup \nwest(v,\infty)$
	    \EndFor
	    \State $\fP \gets \emptyset$ \Comment{$\fP$ is the set of pseudo-terminals.}
	    \For{each $v=(i,j)\in\{V(G)-T(G)\mid \deg_H(v)>0\}$}
	        \If{$(\exists\ i'$ such that $(i',j)\in T(G))$\textbf{ or }$(\exists\ j'$ such that $(i,j')\in T(G))$}
	            \State $\fP\gets \fP\cup\{v\}$
	        \EndIf
	    \EndFor
	    \State \Return $H,\fP$
    \end{algorithmic}\caption{$\textsc{PseudoTerms}(G)$}
    \end{algorithm}
	
	Here is a brief explanation of $\textsc{PseudoTerms}()$. There are two \textbf{for} loops in the algorithm. The first \textbf{for} loop adds edges to $H$. Thus, vertices untouched by the first \textbf{for} loop are isolated in $H$. The second \textbf{for} loop considers non-terminal vertices touched by $H$, and converts them to pseudo-terminals if they have a terminal vertex in their row or column (see~\autoref{rowcol}).
	
	\begin{Definition} \label{rowcol}
		Let $i_0\in[n],j_0\in[n]$.
		\begin{itemize}
		    \item We say that $V(i_0,*)=\{(i_0,j)\mid j\in[n]\}$ is the set of vertices in the column $i_0$ (or $i_0$'s column), and $V(*,j_0)=\{(i,j_0)\mid i\in[n]\}$ is the set of vertices in the row $j_0$ (or $j_0$'s row).
		    \item $G(i_0,*)$ is the subgraph of $G$ induced by the vertex set $V(i_0,*)$, and $G(*,j_0)$ is the subgraph of $G$ induced by the vertex set $V(*,j_0)$.
		\end{itemize}
		Note that $G(i_0,*)$ and $G(*,j_0)$ are interval graphs.
	\end{Definition}
	
	Before we proceed to the $\textsc{StraightPaths}()$ algorithm, let us calculate the number of branching vertices in $H$.
	
	\begin{Claim}
	$H$ has $O(k^2)$ branching vertices.
	\end{Claim}
	\begin{proof}
	Each terminal contributes one \textit{\texttt{northwest}} and one \textit{\texttt{northeast}} path, so that there are $2k$ paths in all. Since these paths are greedy and no two of them are anti-parallel, there is a unique path between any two fixed vertices. Thus each pair of paths can branch in at most two vertices, and since there are $O(k^2)$ pairs of paths, $H$ has $O(k^2)$ branching vertices.
	\end{proof}
	
	\begin{Claim} \label{k2k}
	$G(i_0,*)$ has at most $k$ pseudo-terminals for each $i_0\in [n]$, and $G(*,j_0)$ has at most $2k$ pseudo-terminals for each $j_0\in [n]$.
	\end{Claim}
	\begin{proof}
	Fix $i_0\in[n],j_0\in[n]$ and a terminal $t\in T$. Note that
	\begin{align*}
	|V(i_0,*)\cap(\neast(t,\infty)\cup\nwest(t,\infty))|\leq 1\\
	|V(*,j_0)\cap(\neast(t,\infty)\cup\nwest(t,\infty))|\leq 2
	\end{align*}
	
	In other words, the set of $t$-\textit{\texttt{northwest}} and $t$-\textit{\texttt{northeast}} paths have at most one vertex in common with a fixed column $i$, and at most two vertices in common with a fixed row $j$. Since there are $k$ terminals, this completes the proof.
	\end{proof}
	
	\subsection{Adding Straight Paths}
	 
	The two algorithms $\textsc{PseudoTerms}()$ and $\textsc{StraightPaths}()$ work with each other as follows. For every two terminals $t_i$ and $t_j$, a shortest path between them $P(t_i,t_j)$ is achieved by a combination of two paths $P(t_i,p_{i,j})$ and $P(p_{i,j},t_j)$, where $p_{i,j}\in \fP$ is a pseudo-terminal. Paths of the form $P(t_i,p_{i,j})$ are constructed by $\textsc{PseudoTerms}()$ and paths of the form $P(p_{i,j},t_j)$ are constructed by $\textsc{StraightPaths}()$.
	 
	Once $\textsc{PseudoTerms}()$ has converted the required non-terminal vertices to pseudo-terminals, we need only concern ourselves with graphs of this form. Straight paths are paths that travel either horizontally from left to right in a fixed row, or vertically from bottom to top in a fixed column. We present the algorithm $\textsc{StraightPaths}()$ below.
	
    \begin{algorithm}
	\begin{algorithmic}[1]
	    \For{each $i\in [n]$}
	        \If{$T(i,*)\neq\emptyset$} \Comment{$T(i,*)\triangleq T(G)\cap G(i,*), \fP(i,*)\triangleq\fP\cap G(i,*)$.}
	            \State $H\gets H\cup\textsc{IntervalPaths}(G(i,*),T(i,*), \fP(i,*)\cup T(i,*))$
	        \EndIf
	    \EndFor
	    \For{each $j\in [n]$}
	        \If{$T(*,j)\neq\emptyset$} \Comment{$T(*,j)\triangleq T(G)\cap G(*,j), \fP(*,j)\triangleq\fP\cap G(*,j)$.}
	            \State $H \gets H\cup\textsc{IntervalPaths}(G(*,j),T(*,j), \fP(*,j)\cup T(*,j))$
	        \EndIf
	    \EndFor
	    \State \Return $H$
    \end{algorithmic}\caption{$\textsc{StraightPaths}(G,H,\fP)$}
    \end{algorithm}
    
    Note that $\textsc{StraightPaths}()$ crucially uses the algorithm $\textsc{IntervalPaths}()$ as a subroutine. $\textsc{IntervalPaths}(G,S,T)$ takes an interval graph $G$ with source and target sets $S\subseteq T\subseteq V(G)$ as input, and returns a distance-preserving subgraph of $G$ with $O(|S|\cdot|T|)$ branching vertices. \autoref{intpaths} (which we will prove shortly) shows that this can be done in polynomial time. First let us complete the proof of~\autoref{main2} (a) assuming~\autoref{intpaths}.
    
    \begin{proof}[Proof of~\autoref{main2} (a)] It has already been established that the final subgraph $H$ returned by the algorithm $\textsc{StraightPaths}()$ is distance-preserving for $G$. We will now prove that $H$ has $O(k^2)$ branching vertices.
    
    Fix a column $i\in[n]$. Let $\row_i=|T(i,*)|, \pseudo_i=|\fP(i,*)|$. In other words, $\row_i$ is the number of terminals in column $i$, and $\pseudo_i$ is the number of pseudo-terminals in row $i$. Using the union bound, we get $|\fP(i,*)\cup T(i,*)|\leq \row_i+\pseudo_i$. Thus, $\textsc{IntervalPaths}(G(i,*),T(i,*), \fP(i,*)\cup T(i,*))$ returns a distance-preserving subgraph of $G(i,*)$ with $c(\row_i(\row_i+\pseudo_i))$ branching vertices (\autoref{intpaths}), where $c$ is a constant. Therefore, the total number of branching vertices at the completion of step 5 of $\textsc{StraightPaths}()$ is given by
    \begin{align*}
        \sum_{i=1}^{n}c(\row_i(\row_i+\pseudo_i))&=c\sum_{i=1}^{n}{\row_i}^2+c\sum_{i=1}^{n}\row_i\pseudo_i\\
        &\leq c\left[\sum_{i=1}^{n}\row_i\right]^2+c\sum_{i=1}^{n}\row_i k\qquad &&\because\,\pseudo_i\leq k\ (\autoref{k2k})\\
        &=c\cdot k^2+ck\cdot k &&\because\,\sum_{i=1}^{n}\row_i=k\\
        &=2ck^2.
    \end{align*}
    Thus, $H$ has an additional $O(k^2)$ branching vertices at the completion of step 5 of the algorithm $\textsc{StraightPaths}()$. Similarly, $H$ has an additional $O(k^2)$ branching vertices at the completion of step 10 of $\textsc{StraightPaths}()$.
    
    Finally, we need to account for branching vertices created by the intersections of paths added in step 3 and paths added in step 8. Note that $V(i,*)$ and $V(*,j)$ have exactly one vertex in common for every $(i,j)\in [n]\times [n]$. Since there are at most $k$ columns with $T(i,*)\neq\emptyset$ and at most $k$ rows with $T(*,j)\neq\emptyset$, the number of branching vertices added is at most $k^2$. This completes the proof.
    \end{proof}
    
    Let us now explain the algorithm $\textsc{IntervalPaths}()$. This is just a one-step algorithm. $\textsc{IntervalPaths}(I,S,T)$ takes an interval graph $I$ with source and target sets $S,T$ as inputs, and returns $J$, which is simply a union of greedy shortest paths, as defined below. $$J=\cup\left\{\greedyI(u,v)\mid (u,v)\in (S\times T)\cup ((T-S)\times S)\right\}$$ Clearly, $J$ can be computed in polynomial time. We will now prove that $J$ is distance-preserving for $I$ (that is, $d_J(u,v)=d_I(u,v)$ for all $(u,v)\in S\times T$) and $J$ has $O(|S|\cdot|T|)$ branching vertices.
    
    \begin{Theorem} \label{intpaths}
    Let $p$ and $q$ be positive integers such that $p\leq q$. Given an interval graph $I$ with source and target sets $S\subseteq T\subseteq V(G)$ such that $|S|=p,|T|=q$, a distance-preserving subgraph of $I$ with $O(pq)$ branching vertices can be computed in polynomial time.
    \end{Theorem}
	\begin{proof}
	We will construct $J$ systematically, in a way that it will be convenient to count the number of branching vertices in it. Let $T=\{t_1,t_2,\ldots,t_q\}$ (note that $p$ of these are source vertices) such that $t_i<t_{i+1}$ for every $1\leq i<q$. Let $$J'=\displaystyle\bigcup_{i=1}^{q-1}\greedyI(t_i,t_q).$$ Thus, $J'$ can be thought of as a shortest path tree rooted at $t_q$, but constructed in the reverse direction. Since $J'$ is a tree with at most $q$ leaves, it has at most $q-2$ branching vertices~\footnote{Every tree $T$ with exactly $q$ vertices of degree $1$ has at most $q-2$ vertices of degree at least $3$.}. Also, $J'$ preserves distances from all vertices of $T$ to $t_q$. We will now add $O(pq)$ more edges to $J'$ to produce a $J$ which preserves distances between all pairs of vertices in $S\times T$. 
	
	Fix $1\leq i<j<q$. Let $b_{i,j}$ be the last vertex on the greedy path $\greedyI(t_i,t_j)$ just before $t_j$. Then, $\greedyI(t_i,t_j)=\greedyI(t_i,b_{i,j})\cup (b_{i,j},t_j)$. Note that $\greedyI(t_i,b_{i,j})\subseteq\greedyI(t_i,t_q)$ (since they are greedy paths from the same source, and $b_{i,j}\leq t_q$). Thus, the addition of a single edge $(b_{i,j},t_j)$ to $J'$ preserves the distance between $t_i$ and $t_j$. Let $$J=J'\cup\left[\displaystyle\bigcup_{\substack{1\leq i<j<q,\\t_i\in S,t_j\in T}}(b_{i,j},t_j)\right]\cup\left[\displaystyle\bigcup_{\substack{1\leq i<j<q,\\t_i\in T,t_j\in S}}(b_{i,j},t_j)\right].$$ Thus, $J$ has at most $|S|\cdot|T|$ more vertices than $J'$, and since each edge can contribute at most two additional branching vertices, $J$ has at most $(q-2)+2pq=O(pq)$ branching vertices in all.
	\end{proof}

	\subsection{Proof of the Lower Bound for Bi-interval Graphs}

	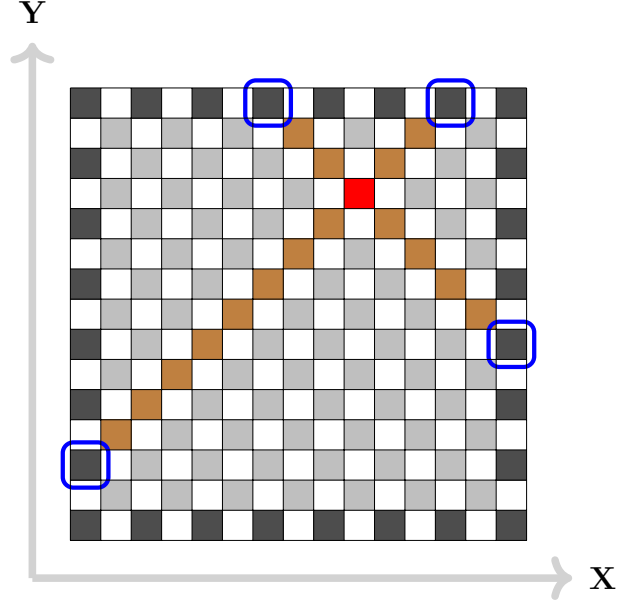
\begin{figure}
	\begin{center}
	\begin{tikzpicture}
    
	\foreach \x in {0,2,...,14}
	    \foreach \y in {0,2,...,14}
	        \fill[fill=lightgray] (0.4*\x,0.4*\y) rectangle +(0.4,0.4);
	\foreach \x in {1,3,...,13}
	    \foreach \y in {1,3,...,13}
	        \fill[fill=lightgray] (0.4*\x,0.4*\y) rectangle +(0.4,0.4);
	\foreach \i in {0,1,...,12}
	    \fill[fill=brown, opacity=1] (0.4*\i,0.4*\i+0.8) rectangle +(0.4,0.4);
    \foreach \i in {4,5,...,12}
	    \fill[fill=brown, opacity=1] (0.4*\i+0.8,-0.4*\i+7.2) rectangle +(0.4,0.4);
	    
    \foreach \x in {0,2,...,12}
	{
	    \fill[fill=black!70] (0.4*\x,0) rectangle +(0.4,0.4);
	    \fill[fill=black!70] (0,0.4*\x) rectangle +(0.4,0.4);
	    \fill[fill=black!70] (0.8+0.4*\x,5.6) rectangle +(0.4,0.4);
	    \fill[fill=black!70] (5.6,0.8+0.4*\x) rectangle +(0.4,0.4);
	}
	
	\fill[fill=black!70] (0,5.6) rectangle +(0.4,0.4);
	\fill[fill=black!70] (5.6,0) rectangle +(0.4,0.4);
	    
	\draw[blue, ultra thick, rounded corners] (2.3,5.5) rectangle +(0.6,0.6);
	\draw[blue, ultra thick, rounded corners] (4.7,5.5) rectangle +(0.6,0.6);
	\draw[blue, ultra thick, rounded corners] (-0.1,0.7) rectangle +(0.6,0.6);
	\draw[blue, ultra thick, rounded corners] (5.5,2.3) rectangle +(0.6,0.6);
	
	\fill[red, fill=red, opacity=1] (0.4*9,0.4*11) rectangle +(0.4,0.4);
	
	\draw[step=0.4, ultra thin] (0,0) grid (6,6);
	
	\draw [->, gray!35, line width=1mm] (-0.5,-0.5) to (6.6,-0.5);
    \node at (7,-0.5) {\large{\bf X}};
    \draw [->, gray!35, line width=1mm] (-0.5,-0.5) to (-0.5,6.6);
    \node at (-0.5,7) {\large{\bf Y}};
	
	\end{tikzpicture}\caption{An instance of $\Gdiag$. The black vertices that lie on the boundary are the terminals. Every internal black vertex (for example, the one coloured red here) lies on two \emph{unique} diagonal shortest paths in $\Gdiag$.}\label{fig:chess}
	\end{center}
	\end{figure}

	In this section, we exhibit a bi-interval graph $\Gdiag$ for which every distance-preserving subgraph has $\Omega(k^2)$ branching vertices, thereby proving~\autoref{main2} (b), restated below.
	
	\begin{Theorem} \label{main2bagain}
	For every $k\geq 4$, there exists a bi-interval graph $\Gdiag$ on $k$ terminals such that every distance-preserving subgraph of $\Gdiag$ has $\Omega(k^2)$ branching vertices.
	\end{Theorem}
	
	It is convenient to follow~\autoref{fig:chess} while reading this proof. Let $\Gdiag$ be a $k \times k$ king's graph (see~\autoref{king}). Colour the squares black and white as the squares of a chess board are ordinarily coloured. Declare the black squares on the boundary to be the terminals. Now, the crucial observation is that for any two terminals that lie on a common (black) diagonal, the diagonal itself is the \emph{unique} shortest path between them. Thus, any distance-preserving subgraph of $\Gdiag$ must include all squares that lie on the intersection of two black diagonals. In particular, all black vertices \emph{not} on the boundary have degree at least $4$ in any distance-preserving subgraph of $\Gdiag$, and are thus branching vertices.
	
    Hence, $\Gdiag$ is a bi-interval graph with $O(k)$ terminals for which every distance-preserving subgraph has $\Omega(k^2)$ branching vertices. This completes the proof of~\autoref{main2} (b).
    
    \subsection{Anti-parallel Greedy Paths in Bi-interval Graphs}
	
	Two paths in bi-interval graphs are called anti-parallel if they are greedy shortest paths in opposite directions. As we shall see shortly, the inclusion of anti-parallel paths can be disastrous to prove any reasonable upper bound for distance-preserving subgraphs of bi-interval graphs.
	
	More specifically, we will show that any algorithm that uses only anti-parallel paths to compute a distance-preserving subgraph of a bi-interval graph cannot achieve an $O(k^2)$ upper bound on the number of branching vertices. In fact, even including just one pair of anti-parallel paths can lead to a large (unbounded in $k$) number of branching vertices. Although we prove this  for \textit{\texttt{east}} versus \textit{\texttt{west}} paths in interval graphs, the proof can be easily extended to \textit{\texttt{northeast}} versus \textit{\texttt{southwest}} (or \textit{\texttt{northwest}} versus \textit{\texttt{southeast}}) paths in bi-interval graphs as well.
	
	\begin{Theorem} \label{antiparallel}
	There exists an interval graph $X$ on $O(n)$ vertices with two terminals $u$ and $v$ such that the number of branching vertices in  is $\Omega(n)$.
	\end{Theorem}
	
	\begin{proof} Let $n$ be a multiple of $4$. Fix $\varepsilon=0.01$ and $\delta=0.1$. Let the set of intervals of $X$ be $\cI=\displaystyle\bigcup_{1\leq k\leq 6}\cI_k$ and the set of terminals be $T=\displaystyle\bigcup_{4\leq k\leq 6}\cI_k$.
	\begin{align*}
	\cI_1&=\{(i-\delta+\varepsilon,i+\delta+\varepsilon)\mid 1\leq i\leq n\},
	&\cI_2&=\{(i-\delta-\varepsilon,i+\delta-\varepsilon)\mid 1\leq i\leq n\},\\
	\cI_3&=\{(i+\varepsilon,i+1-\varepsilon)\mid 1\leq i\leq n-1\},
	&\cI_4&=\{(i+\varepsilon,i+1-\varepsilon)\mid i=0,n\},\\
	\cI_5&=\{(i+\delta,i+2\delta)\mid i=n/4,n/2,3n/4\},
	&\cI_6&=\{(i+1-\delta,i+1-2\delta)\mid i=n/4,n/2,3n/4\}.
	\end{align*}
	Let $u=(\varepsilon,1-\varepsilon)\in\cI_4$ and $v=(n+\varepsilon,n+1-\varepsilon)\in\cI_4$. If we are restricting ourselves to only using greedy paths, then it is easy to see that a $u$-\textit{\texttt{east}} greedy path is required to cover terminals of $\cI_5$ from $u$, and a $v$-\textit{\texttt{west}} greedy path is required to cover terminals of $\cI_6$ from $v$.
	
	Note that a $u$-\textit{\texttt{east}} path uses intervals from $\cI_2$ and $\cI_3$ alternately, and a $v$-\textit{\texttt{west}} path uses intervals from $\cI_1$ and $\cI_3$ alternately. Thus, every vertex of $\cI_3$ has degree $4$, and the number of branching vertices in $\east(u,\infty)\cup\west(v,\infty)$ is $n-1$. Since $|T|=8$, the number of branching vertices cannot be upper bounded in terms of the number of terminals.
	\end{proof}

	
	

	\section{Conclusion}
	
	Finding an all-pairs distance-preserving subgraph with the minimum number of branching vertices is $\NP$-complete for general graphs, and a single-source distance-preserving subgraph is in $\P$ for interval graphs. However, the question remains: what about all-pairs distance-preserving subgraphs for interval graphs? Is it also in $\P$? Is it $\NP$-complete? If so, is it fixed-parameter tractable (with parameter $k$, the number of terminals)? This question is wide open and any progress toward solving it would be interesting.

	
	

	\subsection*{Acknowledgments}
	
	We would like to thank Suhail Sherif for the ``unique diagonal paths'' idea (\autoref{fig:chess}) that led to the lower bound for bi-interval graphs. 

	
	

	\bibliographystyle{alpha}
	\bibliography{ThirdPaper.bib}

	
	


\end{document}